\documentclass{llncs}
\usepackage{microtype}
\pdfoutput=1

\usepackage{amsmath}
\usepackage{amssymb}

\usepackage{graphicx}
\usepackage{epstopdf}

\usepackage{pgf}
\usepackage{tikz}
\usetikzlibrary{automata,calc}
\usepackage{color}
\definecolor{bluestate}{RGB}{135,206,250}

\usepackage{algorithm}
\usepackage{algorithmicx}
\usepackage{algpseudocode}

\usepackage{booktabs} % for \toprule etc
\usepackage{wrapfig}  % for the table

\def\A{\mathcal{A}}

\def\ns{\star}
\def\prod{\mathbf{A}}
\def\t{\mathbf{t}}
\def\s{\mathbf{s}}
\def\U{\mathcal{U}}
\newcommand{\inp}[1]{{}^\bullet\!{#1}}
\newcommand{\out}[1]{{#1}^\bullet}
\def\N{\mathcal{N}}

\def\S{\mathcal{S}}
\def\co{co}
\def\tco{coc}
\def\conf{\mathrel{\#}}

\def\eqrel{\equiv}

\def\scause{\ll}
\newcommand{\st}[1]{\mathbf{St}(#1)}
\newcommand{\past}[1]{[#1]}
\newcommand{\ma}[1]{M(#1)}
\def\ocseq{w}

\newcommand{\Tr}[1]{{\it Tr}(#1)}
\newcommand{\Coni}[2]{{\it Ico}_{#1}(#2)}
\renewcommand{\S}{{\cal S}}
\def\R{\mathbb{R}_+}
\def\trace{tr}

\def\blocks{\vdash_{\!\N}}
\def\rblocks{\mathord{\blocks}}
\def\ind#1{\mathit{ind}(#1)}
\def\parb#1{\mathit{par}(#1)}
\def\worklist{\mathit{worklist}}

\def\by#1{\mathop{{\hbox{\setbox0=\hbox{$\scriptstyle{#1\quad}$}{$\buildrel{\quad\scriptstyle{#1}\quad}\over{\hbox to \wd0{\rightarrowfill}}$}}}}}

\title{Computation of summaries using net unfoldings}
\author{Javier Esparza\inst{1}, Lo\"ig Jezequel\inst{2}, and Stefan Schwoon\inst{3}}
\institute{Institut f\"ur Informatik, Technische Universit\"at M\"unchen, Germany
\and
ENS Cachan Bretagne, Rennes, France
\and
LSV, ENS Cachan \& CNRS, INRIA Saclay, France}

\begin{document}

\maketitle

\begin{abstract}
We study the following {\em summarization problem}: given a parallel
composition $\prod = \A_1 \parallel \ldots \parallel \A_n$ of labelled transition systems
communicating with the environment through a distinguished component $\A_i$, 
efficiently compute a {\em summary } $\S_i$ such that $\mathbf{E} \parallel \prod$ and $\mathbf{E} \parallel \S_i$ are trace-equivalent for every environment $\mathbf{E}$. While $\S_i$ 
can be computed using elementary automata theory, the resulting algorithm suffers from
the state-explosion problem. We present a new, simple but subtle algorithm based on net unfoldings, 
a partial-order semantics, give some experimental results using an implementation on top of \textsc{Mole}, and show that our algorithm can handle {\em divergences} and compute {\em weighted summaries} with minor modifications.

\end{abstract}

\section{Introduction}

We address a fundamental problem in automatic compositional verification. Consider a 
parallel composition $\prod = \A_1 \parallel \ldots \parallel \A_n$ of processes,
modelled as labelled transition systems, which is itself part of a larger system $\mathbf{E} \parallel \prod$ for some environment $\mathbf{E}$.
Assume that $\A_i$ is the interface of $\prod$ with the environment, i.e., 
$\prod$ communicates with the outer world only through actions 
of $\A_i$. The task consists in computing a new interface 
$\S_i$ with the same set
of actions as $\A_i$ such that $\mathbf{E} \parallel \prod$ and $\mathbf{E} \parallel \S_i$ 
have the same behaviour. In other words, the environment $E$ cannot distinguish between 
$\prod$ and $\S_i$. Since
$\S_i$ usually has a much smaller state space than $\prod$ (making
$\mathbf{E} \parallel \prod$ easier to analyse) we call it a {\em summary}.

We study the problem in a CSP-like setting \cite{Hoare85}:
parallel composition is by rendez-vous, and the behaviour of a transition
system is given by its trace semantics. 

It is easy to compute $\S_i$ using elementary automata theory:
we first compute the transition system of $\prod$, whose states  
are tuples $(s_1, \ldots, s_n)$, where $s_i$ is a state of $\A_i$. Then we
hide all actions except those of the interface, i.e., we 
replace them by $\varepsilon$-transitions ($\tau$-transitions in CSP terminology). 
We can then eliminate all $\varepsilon$-transitions using standard algorithms, and, if
desired, compute the minimal summary by applying e.g. Hopcroft's algorithm.
The problem of this approach is the state-space explosion: the number of states of 
$\prod$ can grow exponentially in the number of sequential components. While this
is unavoidable in the worst case (deciding whether $\S_i$ has an empty set of traces 
is a PSPACE-complete problem, and the minimal summary $\S_i$ may be exponentially larger than 
$\A_1, \ldots, \A_n$ in the worst case, see e.g. \cite{HarelKV97}) the combinatorial explosion happens 
already in trivial cases: if the components $\A_1, \ldots, \A_n$ do not communicate at all, 
we can obviously take $\S_i=\A_i$, but the algorithm we have just 
described will need 
exponential time and space.

We present a technique to palliate this problem based on 
net unfoldings (see e.g. \cite{Esparza08}). Net unfoldings are a partial-order semantics for 
concurrent systems, closely related to event structures~\cite{Winskel11}, that provides very 
compact representations of the state space for systems with a high degree of concurrency. 
Intuitively, an unfolding is the extension to parallel compositions
of the notion of unfolding a transition system into a tree. The unfolding
is usually infinite. We show how to algorithmically construct a finite 
prefix of it from which the summary can be easily extracted. The algorithm 
can be easily implemented re-using many 
components of existing unfolders like \textsc{Punf}~\cite{Punf}
and \textsc{Mole}~\cite{Mole}.
However, its correctness proof is surprisingly subtle. This proof is 
the main contribution of the paper. 
However, we also evaluate the algorithm on some classical benchmarks \cite{Cor96}.
We then show that -- with minor modifications -- the algorithm can be extended so that the summary obtained contains information about the possible divergences, that is whether or not after a given finite trace of the interface $\A_i$ it is possible that $\prod$ evolves silently forever (i.e. without using any action of $\A_i$).
And finally, we show how to extend the algorithm to deal with weighted systems: $\S_i$ then also gives for each of its finite traces the minimum cost in $\prod$ to execute this trace.

{\bf Related work.} The summarization problem has been extensively studied in an 
interleaving setting (see e.g. \cite{GrafS90,Valmari96,zaraket2005scalable}), in which one 
first constructs the transition system of $\prod$ and then reduces it. 
We study it in a partial-order setting.

Net unfoldings, and in general partial-order semantics, have been used to solve many analysis problems: deadlock \cite{Mcmillan95,Khomenko00}, 
reachability and model-checking questions \cite{Esparza96,Couvreur00,Khomenko03b,Esparza08,Baldan12}, diagnosis 
\cite{Fabre05}, and other specific applications \cite{Khomenko06,Hickmott07}. To the best of our knowledge 
we are the first to explicitly study the summarization problem. 

Our problem can be solved with the help of Zielonka's algorithm \cite{Zielonka87,Mukund94,Genest10}, which yields 
an asynchronous automaton trace-equivalent to $\prod$. The projection of this automaton onto the alphabet of $\A_i$ 
yields a summary $\S_i$. However, Zielonka's algorithm is notoriously complicated and, contrary to our algorithm, 
requires to store much additional information for each event \cite{Mukund94}. 
In \cite{Fabre09},  the complete tuple $\S_1, \ldots, \S_n$ is computed -- possibly in a weighted context -- with an iterative message-passing algorithm that 
transfers information between components until a fixed point is reached. However, termination
is only guaranteed when the communication graph is acyclic.

This paper extends~\cite{Esparza13} with proofs and implementation details.

\section{Preliminaries}

\subsection{Transition systems}

A \emph{labelled transition system} (LTS) is a tuple $\A=(\Sigma,S,T,\lambda,s^0)$ where $\Sigma$ 
is a set of \emph{actions}, $S$ is a set of \emph{states}, 
$T \subseteq S \times S$ is a set of \emph{transitions}, $\lambda \colon T \rightarrow \Sigma$
is a {\em labelling function}, and $s^0\in S$ is an \emph{initial state}. 
An $a$-transition
is a transition labelled by $a$. 
We use this definition -- excluding the possibility to have two transitions with different labels between the same pair of states -- for simplicity. 
However, the results presented in this paper would still hold if this possibility was not excluded.
A sequence of transitions $\tau = t_1t_2t_3\dots \in T^*\cup T^\omega$ is an \emph{execution} of $\A$ if there
is a sequence $s_0s_1s_2\dots$ of states such that $t_k=(s_{k-1},s_k)$ 
for every $k$. We write $s_0 \by{\tau}$ (or $s_0\by{\tau} s_n$ when $\tau$ is finite with $t_n$ as last transition). An execution 
is a \emph{history} if $s_0=s^0$. 
A sequence $\sigma=a_1a_2a_3 \ldots \in\Sigma^*\cup\Sigma^\omega$ 
of actions is a \emph{computation} if there is an execution $\tau = t_1t_2t_3\dots$
such that $\lambda(\tau)=\lambda(t_1)\lambda(t_2)\lambda(t_3)\ldots = \sigma$; if $s_0 \by{\tau}$, 
then we also write $s_0 \by{\sigma}$. 
It is a \emph{trace} if{}f there exists such $\tau$ which is an history.
We call $\tau$ a \emph{realization} of $\sigma$.
Abusing language, given an execution $\tau=t_1t_2t_3 \ldots $, we denote by $\trace(\tau)$ the 
computation $\lambda(t_1)\lambda(t_2)\lambda(t_3) \ldots$ (even if it is not necessarily a trace).
The set of traces 
of $\A$ is denoted by $\Tr{\A}$.
Figure~\ref{fig:transitionsystems} shows (on its left) three transition systems. 

\begin{figure}[htbp]
\centering
\includegraphics[scale=0.7]{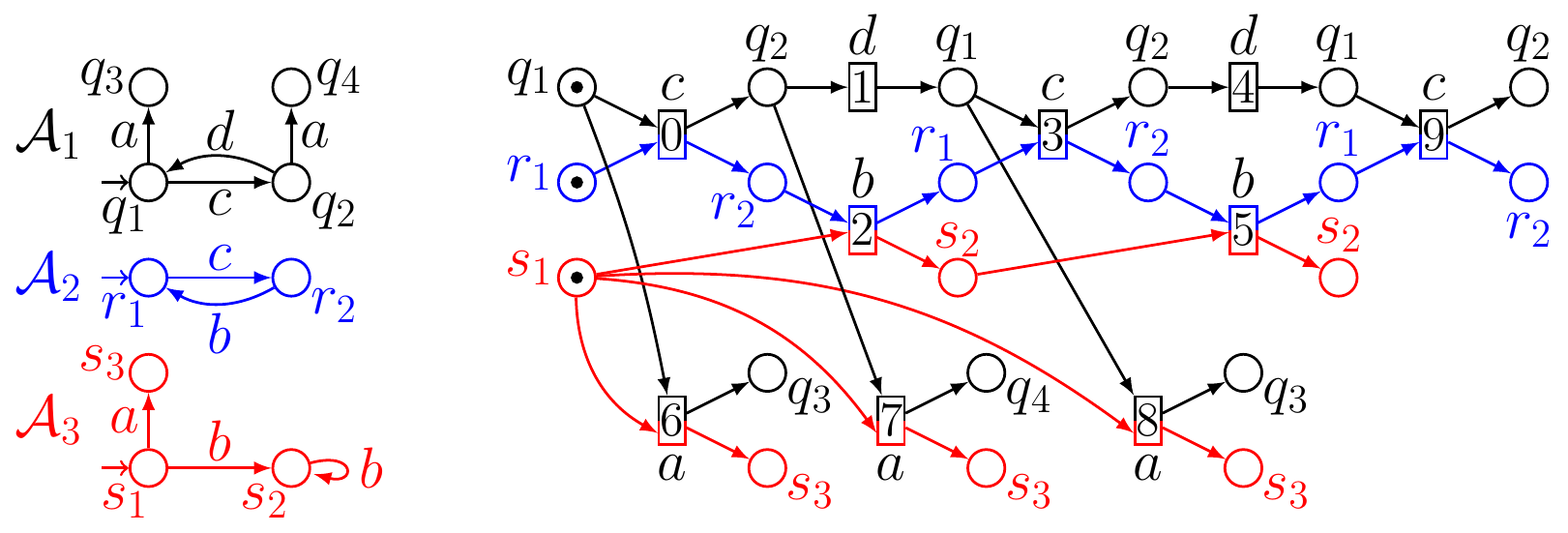}
\caption{Three labeled transition systems (left) and a branching process (right)}\label{fig:transitionsystems}\label{fig:petrinet}\label{fig:unfolding}
\end{figure}

Let $\A_1,\dots,\A_n$ be LTSs where $\A_i=(\Sigma_i,S_i,T_i,\lambda_i,s^0_i)$. 
The {\em parallel composition}  $\prod=\A_1 \parallel \ldots \parallel \A_n$ is the LTS 
defined as follows. The set of actions
is $\mathbf{\Sigma} =\Sigma_1 \cup \ldots \cup \Sigma_n$. The states, called \emph{global states}, are the tuples  
$\s=(s_1,\dots,s_n)$ such that $s_i\in S_i$ for every $i\in\{1..n\}$. 
The \emph{initial global state} is $\s^0=(s_1^0,\dots,s_n^0)$. The transitions, called {\em global transitions},
are the tuples $\t=(t_1, \ldots, t_n)  \in (T_1\cup\{\ns\})\times\dots\times(T_n\cup\{\ns\})\setminus\{(\ns,\dots,\ns)\}$ such that there is an action $a \in \mathbf{\Sigma}$ satisfying 
 for every $i\in\{1..n\}$: if $a \in \Sigma_i$, then $t_i$ is an $a$-transition of $T_i$,
otherwise $t_i = \ns$; the label of $\t$ is the action $a$. 
If $t_i\neq\ns$ we say that $\A_i$ \emph{participates} in $\t$. 
It is easy to see that $\sigma \in \mathbf{\Sigma}^*\cup\mathbf{\Sigma}^\omega$ is a trace of $\prod$ if{}f for every $i \in \{1..n\}$ the projection of $\sigma$ on $\Sigma_i$, denoted by $\sigma_{|\Sigma_i}$, is a trace of $\A_i$.

\subsection{Petri nets}

A \emph{labelled net} is a tuple $(\Sigma, P,T,F,\lambda)$ where $\Sigma$ is a set of {\em actions}, $P$ and $T$ are 
disjoint sets of \emph{places} and \emph{transitions} (jointly called \emph{nodes}), $F\subseteq (P\times T)\cup (T\times P)$ is a set
of {\em arcs}, and $\lambda \colon P \cup T \rightarrow \Sigma$ is a {\em labelling
function}. For $x\in P\cup T$ we denote by $\inp{x}=\{\,y\mid(y,x)\in F\,\}$ and $\out{x}=\{\,y\mid(x,y)\in F\,\}$ the sets of \emph{inputs} and \emph{outputs} of $x$, respectively. A set $M$ of places is called a \emph{marking}.  
A \emph{labelled Petri net} is a tuple $\N=(\Sigma,P,T,F,\lambda,M_0)$ where $(\Sigma,P,T,F,\lambda)$ is a labelled net and $M_0\subseteq P$ is the \emph{initial marking}.
A marking $M$ \emph{enables} a transition $t\in T$ if $\inp{t}\subseteq M$.
In this case $t$ can \emph{occur} or \emph{fire}, leading to the new marking $M'=(M\setminus\inp{t})\cup\out{t}$.
An \emph{occurrence sequence} is a (finite or infinite) sequence of transitions that can occur from $M_0$ in the 
order specified by the sequence. A \emph{trace} is the sequence of labels of an occurrence 
sequence. The set of traces of $\N$ is denoted by $\Tr{\N}$.

\subsection{Branching processes}

The finite \emph{branching processes} of $\prod=\A_1~\parallel~\ldots~\parallel~\A_n$
are labelled Petri nets whose places are labelled with states of $\A_1,\dots,\A_n$, and 
whose transitions are labelled with global transitions of $\prod$. Following tradition, 
we call the places and transitions of these nets {\em conditions} and {\em events}, respectively.
 (Since global transitions are labelled with 
actions, each event is also implicitly labelled with an action.)
We say that a marking $M$ of these nets {\em enables} a global transition 
$\t$ of $\prod$ if for every state $s \in \inp{\t}$ some condition of $M$ is labelled by $s$. 
The set of {\em finite branching processes} of $\prod$ is defined inductively as follows:
\begin{enumerate}
\item A labelled Petri net with conditions $b_1^0, ..., b_n^0$ labelled by $s_1^0, \ldots, s_n^0$, 
no events, and  
with initial marking $\{b_1^0, ..., b_n^0\}$,
is a branching process of $\prod$.
\item Let $\N$ be a branching process of $\prod$ such that some reachable marking of $\N$ enables some  
global transition $\t$. Let $M$ be the subset of conditions of the marking labelled by $\inp{\t}$. 
If $\N$ has no event labelled by $\t$ with $M$ as input set, then the Petri net
obtained by adding to $\N$: a new event $e$, labelled by $\t$; a new condition for every state $s$ of 
$\out{\t}$, labelled by $s$; new arcs leading from each condition of $M$ to $e$, and from $e$ to each 
of the new conditions, is also a branching process of $\prod$. 
\end{enumerate}
Figure~\ref{fig:unfolding} shows on the right a branching process 
of the parallel composition of the LTSs on the left. Events are labelled with their corresponding actions.

The set of all branching processes of a net, finite and infinite, is defined by
closing the finite branching processes under countable unions (after a suitable renaming of conditions and events) \cite{Esparza08}.
In particular, the union of 
all finite branching processes yields the {\em unfolding} of the net, which 
intuitively corresponds to the result of exhaustively adding all extensions in the definition above. 

A {\em trace} of a branching process $\N$ is the sequence of action labels of an occurrence sequence of events of $\N$. In Figure~\ref{fig:unfolding},
firing the events on the top half of the process yields any of the traces 
$cbdcbd$, $cdbcbd$, $cbdcdb$, or $cdbcdb$.
The sets of traces of $\prod$
and of its unfolding coincide.

Let $x,y$ be nodes of a branching process. We say that $x$ is 
a \emph{causal predecessor} of $y$, denoted by $x < y$, if there is a non-empty path 
of arcs from $x$ to $y$; further,
$x\leq y$ denotes that either $x<y$ or $x=y$. If $x\leq y$ or $x\geq y$, then $x$ and 
$y$ are \emph{causally related}.
We say that $x$ and $y$ are \emph{in conflict}, denoted by $x \conf y$, if there is a 
condition $z$ (different from $x$ and $y$) from which one can reach both $x$ and $y$, 
exiting $z$ by different arcs.
Finally, $x$ and $y$ are \emph{concurrent} if they are neither causally related nor 
in conflict.

A set of events $E$ is a \emph{configuration} if it is \emph{causally closed} (that is, 
if $e\in E$ and $e'<e$ then $e'\in E$) and \emph{conflict-free} (that is, for every 
$e,e'\in E$, $e$ and $e'$ are not in conflict). 
The \emph{past} of an event $e$, denoted by $\past{e}$, is the set of events $e'$ such 
that $e'\leq e$ (so it is a configuration). For any event $e$,  we denote by $\ma{e}$ the unique marking reached by 
any occurrence sequence that fires exactly the events of $\past{e}$. 
Notice that, for each component $\A_i$ of $\prod$, $\ma{e}$ contains exactly one condition 
labelled by a state of $\A_i$. We denote this condition by $\ma{e}_i$.
We write $\st{e}=\{\,\lambda(x)\mid x\in \ma{e}\,\}$ and call it the 
\emph{global state reached by $e$}.

\section{The Summary Problem}

Let $\prod=\A_1 \parallel \dots \parallel \A_n$ be a parallel composition with a distinguished
component $\A_i$, called the {\em interface}. An {\em environment} of $\prod$ is any LTS
$\mathbf{E}$ (possibly a parallel composition) that only communicates with $\prod$ through the interface,
i.e, $\Sigma_\mathbf{E} \cap (\Sigma_1 \cup \ldots \cup \Sigma_n) = \Sigma_\mathbf{E} \cap \Sigma_i$. We wish to
compute a {\em summary} $\S_i$, i.e., an LTS with the same actions as $\A_i$ such that 
$\Tr{\mathbf{E} \parallel \prod}|_{\Sigma_\mathbf{E}} = \Tr{\mathbf{E} \parallel \S_i}|_{\Sigma_\mathbf{E}}$ for every 
environment $\mathbf{E}$, where $X|_{\Sigma}$ denotes the projection of the traces of $X$ onto 
$\Sigma$. It is well known (and follows easily from the definitions) that this holds if{}f 
$\Tr{\S_i} = \Tr{\prod}|_{\Sigma_i}$ \cite{Hoare85}. We therefore address the following problem:

\begin{definition}[Summary problem]
Given LTSs $\A_1, \ldots, \A_n$ with interface $\A_i$,
compute an LTS $\S_i$ satisfying $\Tr{\S_i} = \Tr{\prod}|_{\Sigma_i}$, 
where $\prod = \A_1 \parallel \cdots \parallel \A_n$. 
\end{definition}

The problem can be solved by computing the LTS $\prod$, but the size of $\prod$ 
can be exponential in $\A_1, \ldots, \A_n$. So we investigate an unfolding approach. 

The {\em interface projection} $\N_i$ of a branching process $\N$ of $\prod$ onto 
$\A_i$ is the following labelled subnet of $\N$:
(1) the conditions of $\N_i$ are the conditions of $\N$ with labels in $S_i$;
(2) the events of $\N_i$ are the events of $\N$  where $\A_i$ participates;
(3) $(x,y)$ is an arc of $\N_i$ if{}f it is an arc of $\N$ and $(x,y)$ are nodes of $\N_i$.
Obviously, every event of $\N_i$ has exactly one
input and one output condition, and $\N_i$ can therefore be seen as an LTS; thus, we sometimes
speak of the LTS $\N_i$. The interface projection $\N_1$ for the branching process of 
Figure \ref{fig:unfolding} is the subnet given by the black conditions and their input and output 
events, and its LTS representation is shown in the left of Figure \ref{fig:projection}.

\begin{figure}[htbp]
\centering
\includegraphics[scale=0.55]{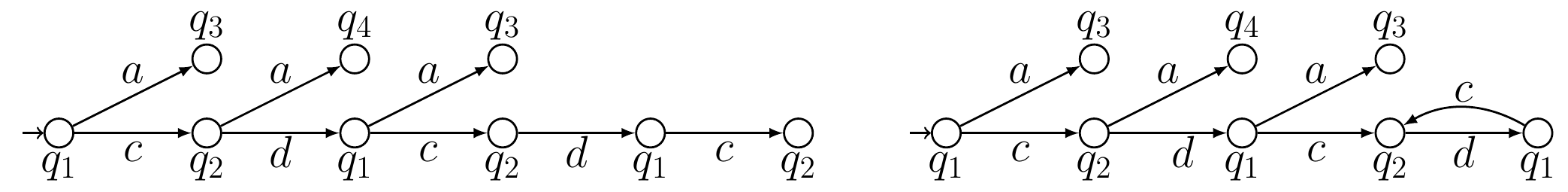}
\caption{Projection of the branching process of Figure~\ref{fig:unfolding} on 
$\A_1$ (left) and a folding (right)}\label{fig:projection}
\end{figure}

The projection $\U_i$ of the full unfolding of $\prod$ onto $\A_i$ clearly 
satisfies $\Tr{\U_i} = \Tr{\prod}_{|\Sigma_i}$; however, $\U_i$ can be
infinite. In the rest of the paper we show how to compute a 
{\em finite} branching process $\N$ and an equivalence relation $\eqrel$ between the 
conditions of $\N_i$ such that the result of {\em folding} $\N_i$ into a finite 
LTS by merging the conditions of each equivalence class yields the desired $\S_i$. 
The {\em folding} of $\N_i$ 
is the LTS whose states are
the equivalence classes of $\equiv$, and every transition $(s,s')$ of $\N_i$ yields a 
transition $([s]_\equiv,[s']_\equiv)$ of the folding.
Figure~\ref{fig:projection} shows on the right
the result of folding the LTS on the left 
when the only equivalence class with more than one member is formed by the two rightmost states labelled by $q_2$.

We construct $\N$ by starting with the branching processes without events and
iteratively add one event at a time.
Some events are marked as {\em cut-offs} \cite{Esparza08}. An event $e$ added to $\N$ becomes 
a cut-off if $\N$ already contains an $e'$, called the {\em companion} of $e$, satisfying 
a certain, yet to be specified \emph{cut-off criterion}. Events with cut-offs
in their past cannot be added. The algorithm 
terminates when no more events can be added. 
The equivalence relation $\equiv$ is determined by the {\em interface cut-offs}: the cut-offs labelled with interface actions. 
If an interface cut-off $e$ has companion $e'$, then we set $M(e)_i\equiv M(e')_i$. 
Algorithm~\ref{algo:unfolding} is pseudocode for the unfolding, where
$Ext(\N,co)$ denotes the 
{\em possible extensions}: the events which can be added to $\N$ without
events from the set $co$ of cut-offs in their past.

\begin{algorithm}[htbp]
\begin{algorithmic}
\State let $\N$ be the unique branching process of $\prod$ without events and let $\co=\emptyset$
\State {\bf While} $Ext(\N,\co)\neq\emptyset$ {\bf do}
\State\hspace*{0.4cm} choose $e$ in $Ext(\N,\co)$ and extend $\N$ with $e$
\State\hspace*{0.4cm} {\bf If} $e$ is a cut-off event {\bf then} let $\co=\co\cup\{e\}$
\State {\bf For every} $e\in\co$ with companion $e'$ {\bf do} merge $[M(e)_i]_\equiv$ and $[M(e')_i]_\equiv$
\end{algorithmic}
\caption{Unfolding procedure for a product $\prod$.}
\label{algo:unfolding}
\end{algorithm}

Notice that the algorithm is nondeterministic: the order
in which events are added is not fixed (though it necessarily respects causal relations).
We wish to find a definition of cut-offs such that the LTS $\S_i$
delivered by the algorithm
is a correct solution to the summary problem. 
Several papers have addressed the problem of defining cut-offs such that the
branching process delivered by the algorithm contains all global states of 
the system (see \cite{Esparza08} and the references therein). 
We first remark that these approaches do not ``unfold enough''.

\paragraph{\bf Standard cut-off condition does not work.}

Usually, an event $e$ is declared a cut-off if the 
branching process already contains an event $e'$ with the 
same global state. If events are added according to an {\em adequate order} \cite{Esparza08}, then the prefix generated by the algorithm is guaranteed to contain occurrence sequences leading to all reachable markings.

We show that with this definition of cut-off even we do not always compute
a correct summary. We do so by showing an example in which {\em independently of the order in which Algorithm~\ref{algo:unfolding} adds events} the summary is always wrong. Consider the parallel composition of Figure~\ref{fig:transitionsystemsapp} with $\A_1$ as interface.

\begin{figure}[htbp]
\centering
\includegraphics[scale=0.7]{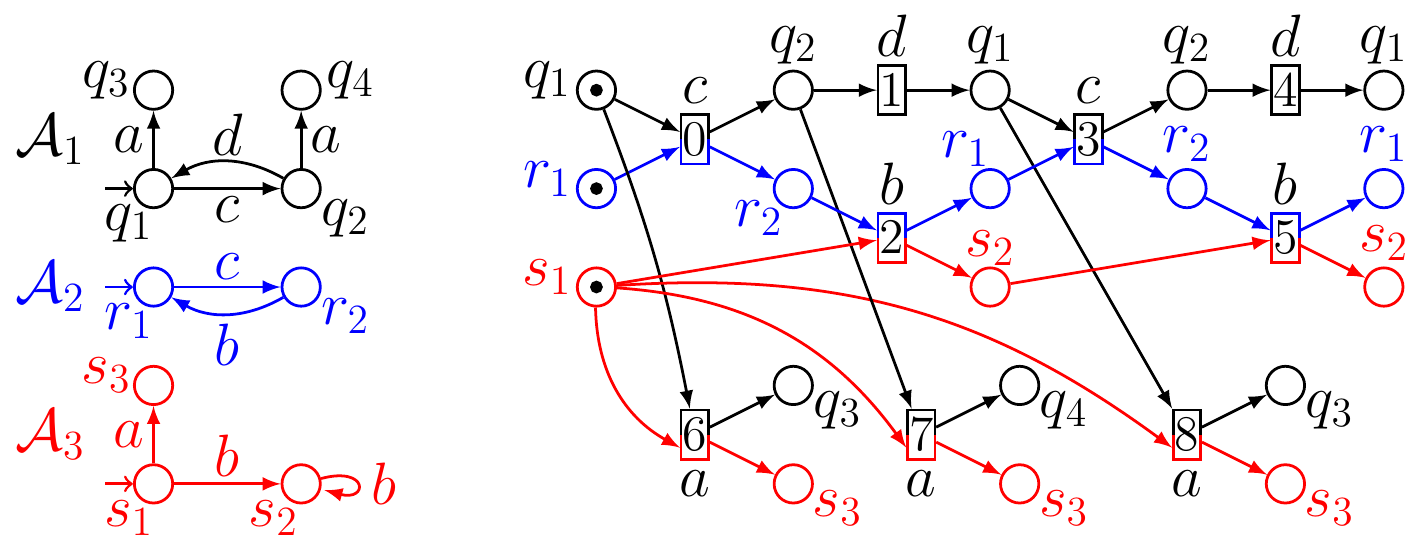}
\caption{Three labeled transition systems (left) and a branching process (right)}\label{fig:transitionsystemsapp}
\end{figure}
  
Independently of the order in which events are added, the branching process $\N$ computed  
by Algorithm \ref{algo:unfolding} is the one shown on the right of Figure~\ref{fig:transitionsystemsapp}. The only cut-off event is $5$, with companion event $2$, for which we 
have  $\st{5}=\{q_2, r_1, s_2\}=\st{2}$.  The interface projection $\N_1$ is the transition system in Figure~\ref{fig:projectionapp}. 

\begin{figure}[htbp]
\centering
\includegraphics[scale=0.55]{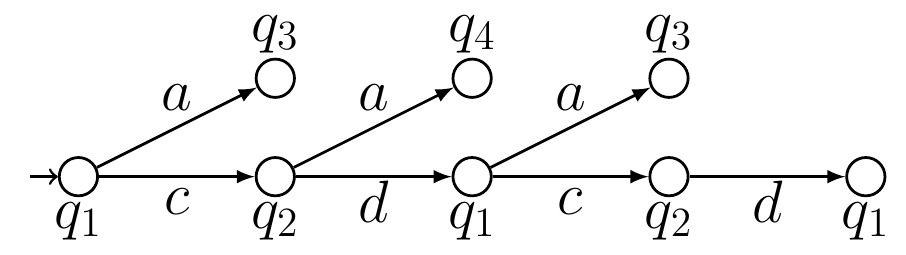}
\caption{Projection of the branching process of Figure~\ref{fig:transitionsystemsapp} on 
$\A_1$}\label{fig:projectionapp}
\end{figure}

Since $\N_1$ does not contain any cut-off, its folding is again
$\N_1$, and since $\Tr{\prod}|_{\Sigma_1}\supseteq cdc(dc)^*$, $\N_1$ is
not a summary.

\section{Two Attempts}

The solution turns out to be remarkably subtle, and so we approach it in a series of steps.

\subsection{First attempt} 

In the following we shall call events in which $\A_i$ participates
{\em $i$-events} for short; analogously, we call {\em $i$-conditions} the 
conditions labelled by states of $\A_i$.

The simplest idea is to declare an $i$-event $e$ a cut-off if the 
branching process already contains another $i$-event $e'$ with $\st{e}=\st{e'}$. Intuitively, the behaviours of the interface after the configurations $[e]$ and $[e']$ is identical, and so we only explore the future of $[e']$.

\begin{quote}
{\bf Cut-off definition 1.} An event $e$ is a cut-off event 
if it is an $i$-event and $\N$ contains an $i$-event $e'$ such that
 $\st{e}=\st{e'}$. 
\end{quote}

It is not difficult to show that this definition is correct for 
{\em non-divergent} systems.

\begin{definition}
A parallel composition $\prod$ with interface $\A_i$ is 
{\em divergent} if some infinite trace of $\prod$ contains only finitely many occurrences of actions of $\Sigma_i$. 
\end{definition}

\begin{theorem}
\label{th:interfacefair}
Let $\prod$ be non-divergent. The instance of Algorithm~\ref{algo:unfolding} with cut-off definition 1 terminates with a finite branching process $\N$, and the folding $\S_i$ of $\N_i$ is a summary of $\prod$.
\end{theorem}

\begin{proof}
Let $\N$ be the branching process constructed by Algorithm~\ref{algo:unfolding}.
Assume $\N$ is infinite (i.e., the algorithm does not terminate). 
Then $\N$ contains an infinite chain $e_1 < e_2 \cdots$ of causally related events~\cite{Khomenko03}. 
Since $\prod$ is non-divergent, the infinite configuration $C=\bigcup_{i=1}^\infty [e_i]$ contains infinitely
many $i$-events. Since the interface $\A_i$ participates in all of them, they are all causally related, 
and so $C$ contains an infinite chain $e_1' < e_2' \ldots$ of causally related $i$-events. Since
$\prod$ has only finitely many global states, the chain contains two $i$-events $e_j'<e_k'$ such that
$\st{e_j'}=\st{e_k'}$. So $e_k'$ is a cut-off, in contradiction with the fact that $e_{k+1}'$ belongs to $\N$. 
So $\N$ is finite, and so Algorithm~\ref{algo:unfolding} terminates.

\vspace{0.2cm}
It remains to prove $\Tr{\S_i}=\Tr{\prod}|_{\Sigma_i}$. We prove both inclusions separately,
but we first need some preliminaries. We extend the mapping $\st{}$ to conditions
by defining $\st{b}=\st{e}$, where $e$ is the unique input event of condition $b$. Since the states of $\S_i$ 
are equivalence classes of conditions of $\N_i$ and, by definition, if $b \equiv b'$ 
then $\st{b} = \st{b'}$, we can extend $\st{}$ further to equivalence classes by defining
$\st{[b]_\equiv} = \st{b}$. 

$\Tr{\S_i} \subseteq \Tr{\prod}|_{\Sigma_i}$. Let $\trace^i$ be a trace of $\S_i$. Then $[b^0]_\equiv \by{\trace^i}$ in $\S_i$,
where $[b^0]_\equiv$ is the initial state of $\S_i$.
By the definition of folding, there exist $\trace^i_{1},\trace^i_{2},\trace^i_{3}, \ldots$ (finite sequences of actions) and pairs $(b_1',b_1),(b_2',b_2),(b_2',b_2), \ldots$ of conditions of $\N_i$ such that
(1) $\trace^i = \trace^i_{1} \trace^i_{2} \trace^i_{3}\ldots $; (2) $b^0=b'_1$; (3) $b_{j}' \by{\trace_{j}^i} b_j$ in $\N_i$ for every $j$; 
and (4) $b_{j-1} \equiv  b_j'$ for every $j$. 

By (3) and the definition of projection, we have
$\st{b_{j}'} \by{\trace_{j}} \st{b_j}$ in $\prod$ for some $\trace_j \in \Sigma^*$ such that $\trace^i_{j}=\trace_j|_{\Sigma_i}$: indeed, if $e$ and $e'$ are the input events of $b_j$ and $b'_{j}$, then $\st{b_j}$ is reachable from $\st{b'_{j-1}}$ by means of any computation $\trace_j$ corresponding to executing the events of $[e] \setminus [e']$, and any
such $\trace_j$ satisfies $\trace^i_{j}=\trace_j|_{\Sigma_i}$. Moreover, by (4) we have $\st{b_{j-1}}=\st{b_j'}$. So we get 
$$\st{b_1'} \by{\trace_{1}} \st{b_2'} \by{\trace_{2}} \st{b_3'} \by{\trace_{3}} \cdots$$
\noindent By (1) and (2) we have $\st{b^0}\by{\trace_1\trace_2\trace_3\dots}$ in $\prod$, and so
$\trace^i=\trace|_{\Sigma_i}\in\Tr{\prod}|_{\Sigma_i}$ with $\trace=\trace_1\trace_2\trace_3\dots$.

$\Tr{\prod}|_{\Sigma_i} \subseteq \Tr{\S_i}$.  Let $\trace$ be a finite or infinite trace of $\prod$.
We prove that there exists a trace $\trace^i$ of $\S_i$ such that $\trace^i=\trace|_{\Sigma_i}$.
For that we prove that for every history $h$ of $\prod$ there exists a history $h^i$ of $\S_i$ 
such that $\trace(h^i)=\trace(h)|_{\Sigma_i}$. 

A finite history $h=\t_1 \ldots \t_k$ is {\em short} if the unique sequence of events of the unfolding
$e_1 \ldots e_k$ such that $\lambda(e_\ell)=\t_\ell$ for every $\ell \in \{1..k\}$ satisfies the
following conditions: $e_\ell \leq e_k$ for every $\ell \in \{1..k\}$, and $e_k$ is an $i$-event.
(The name is due to the fact that, loosely speaking, $h$ is a shortest history in which $e_k$ occurs.)  

We say that a finite or infinite history $h$ is {\em succinct} if there are $h_1,h_2,h_3 \ldots$ such that 
$h = h_1h_2h_3 \ldots$, $|\trace(h_k)|_{\Sigma_i}|=1$ for every $k$, and $h_1 \ldots h_\ell$ is short for every $\ell$.
We call $h_1 h_2 h_3\dots$ the {\em $i$-decomposition} of $h$.
It is easy to see that for every history $h$ of $\prod$ there exists a succinct history $h'$ of $\prod$ with the same 
projection onto $\A_i$ (let $o=o_1o_2o_3\dots$ be the occurrence sequence such that $\lambda(o)=h$, denote by $e_{i_1}e_{i_2}e_{i_3}\dots$ its $i$-events in the order they appear in $o$, then simply take for $h'$ any history with $i$-decomposition $h'_1h'_2h'_3\dots$ such that, for any $\ell$, $h'_1\dots h'_\ell$ is an history corresponding to $[e_{i_\ell}]$).
So it suffices to prove the result for 
succinct histories. 

We prove by induction the following stronger result.
For every succinct history of $\prod$ with $i$-decomposition $h_1h_2h_3 \ldots$ there exist $h_1^i,h_2^i,h_3^i,\dots$ such that for every $k$:
\begin{enumerate}
\item[(a)] $H^i_k=h^i_1\dots h^i_k$ is an history of $\S_i$ such that $\trace(H^i_k)=\trace(h_1\dots h_k)|_{\Sigma_i}$.
\item[(b)] There exists a configuration $C_k$ of $\N$ that contains no cut-offs and such that
$[\ma{C_k}_i]_\equiv$ is the state reached by $H^i_k$.
\end{enumerate}

\noindent{\bf Base case.}
If $k=0$, then $H^i_k$ is the empty history of $\S_i$, take $C_k = \emptyset$.

\noindent{\bf Inductive step.}
Let $H_{k+1}$ be the prefix of $h$ with $i$-decomposition $H_{k+1}=h_1 \ldots h_k h_{k+1}$ (it is a succinct history of $\prod$).
Then $H_k = h_1 \ldots h_k$ is succinct with $i$-decomposition $h_1 \ldots h_k$. 
By induction hypothesis $H_k^i=h_1^i\dots h_k^i$ and some configuration $C_k$
satisfy the conditions above. 

Let $o_{k+1}=e_1 \dots e_m$, where $m = |h_{k+1}|$,  be the only sequence of events 
whose labelling is $h_{k+1}$ and can occur in the order of the sequence from the marking $\ma{C_k}$ 
(this sequence always exists by the properties of $C_k$). 
Two cases are possible.
\begin{enumerate}
\item $o_{k+1}$ contains no cut-off.
In this case $o_{k+1}$ is a sequence of events from $\N$ (because $C_k$ contains no cut-offs). 
Thus, there exists an execution $h_{i,k+1}$ of $\S_i$ from the state $[\ma{C_k}_i]_\equiv$ 
to the state $[\ma{e_{m}}_i]_\equiv$ such that $\trace(h_{i,k+1})=\trace(h_{k+1})|_{\Sigma_i}$.
So we can take $h^i_{k+1}=h_{i,k+1}$.
It remains to choose the configuration $C_{k+1}$. We take $C_{k+1}$ as $C_k \cup \{e_1, \ldots, e_m\}$,
which contains no cut-offs because $C_k$ contains no cut-offs by hypothesis.
\item $o_{k+1}$ contains some cut-off.
Since $h_k$ is succinct, $e_m$ is the only $i$-event of $h_{k+1}$, and the only
maximal event of $\{e_1, \ldots, e_m\}$ w.r.t. the causal relation. Since
only $i$-events can be cut-offs, $e_m$ is a cut-off, and the only cut-off among the events of
$o_{k+1}$. So $o_{k+1}$ is a sequence of events from $\N$ whose last event is a cut-off.
Further, by the maximality of $e_m$, the marking reached by $o_{k+1}$ is $\ma{e_m}$.
By the definition of folding, $\S_i$ has an execution $h_{i,k+1}$ from the state 
$[\ma{C_{k}}_i]_\equiv$ to the state $[\ma{e_m}_i]_\equiv$ such that 
$\trace(h_{i,k+1})=\trace(h_{k+1})|_{\Sigma_i}$. As above, this allows to take $h^i_{k+1}=h_{i,k+1}$.

It remains to choose the configuration $C_{k+1}$. We cannot take 
$C_{k+1} = C_{k} \cup \{e_1, \ldots, e_m\}$, because then $C_{k+1}$ would contain cut-offs. 
So we proceed differently. We choose $C_{k+1}= [e_m']$, where $e_m'$ is the companion of $e_m$. 
Since $e_m'$ is not a cut-off, $C_{k+1}$ contains no cut-offs. Moreover, since the marking 
reached by $o_{k+1}$ is $\ma{e_m}$, we have that $[\ma{C_{k+1}}_i]_\equiv$ is the state reached by 
$H_{k+1}^i$.
\end{enumerate}
\end{proof}

The system of Figure~\ref{fig:transitionsystems} is non-divergent. 
Algorithm \ref{algo:unfolding} computes the branching process 
on the right of Figure~\ref{fig:transitionsystems}. The only cut-off is
event $9$ with companion $3$. The folding is shown in Figure~\ref{fig:projection} (right) and is a correct summary.
However, cut-off definition 1 {\em never} works if $\prod$ is divergent
because the unfolding procedure does not terminate. Indeed, if the system 
has divergent traces then we can easily construct an infinite firing sequence of the unfolding such that none of the finitely many $i$-events 
in the sequence is a cut-off. Since no other events can be cut-offs, Algorithm~\ref{algo:unfolding} adds all events of the sequence. This occurs for 
instance for the system of Figure~\ref{fig:needconc} with interface $\A_1$, where the occurrence sequence of the unfolding for the trace $\mathit{i}(\mathit{fcd})^\omega$ contains no cut-off.

\subsection{Second attempt} 
To ensure termination for divergent systems, we extend the definition of cut-off. For this, we define for each event $e$ its {\em $i$-predecessor}. Intuitively, the $i$-predecessor of an event $e$ is the last condition that $e$ ``knows'' has been
reached by the interface. 

\begin{definition}
The 
{\em $i$-predecessor} of an event $e$, denoted by $ip(e)$, is the condition $M(e)_i$.
\end{definition}

Assume now that two events $e_1 < e_2$, neither of them interface event, satisfy $ip(e_1)=ip(e_2)$
and $\st{e_1}=\st{e_2}$. Then any occurrence sequence $\sigma$ that executes 
the events of the set $[e_2] \setminus [e_1]$ leads from a marking to itself and 
contains no interface events. So $\sigma$ can be repeated infinitely often,
leading to an infinite trace with only finitely many interface actions. It is therefore
plausible to mark $e_2$ as cut-off event, in order to avoid this infinite repetition.

\begin{quote}
{\bf Cut-off definition 2.} An event $e$ is a cut-off if 
\begin{enumerate}
\item[(1)] $e$ is an $i$-event, and $\N$ contains an $i$-event $e'$ with $\st{e}=\st{e'}$, or
\item[(2)] $e$ is not an $i$-event, and some event $e'< e$ satisfies $\st{e}=\st{e'}$ and $ip(e)=ip(e')$.
\end{enumerate}
\end{quote}

We give 
an example showing that this natural definition does not work: the algorithm always terminates
but can yield a wrong result. Consider the parallel composition at the left of
Figure \ref{fig:concneeded}, with interface $\A_1$. Clearly
$\Tr{\A}|_{\Sigma_1} = \Tr{A_1} = iab^*e$. For any strategy the algorithm generates the branching 
process $\N$ at the top right of the figure (without the dashed part). $\N$ has two cut-off events: the interface event $6$,
which is of type (1), and event $8$, a non-interface event, of type (2). Event $6$ has $5$ as companion,
with $\st{5}=\st{6}= \{q_2,r_2,s_2\}$. Event $8$ has $0$ as companion, with $\st{0}=\{q_1,r_1,s_1\}=\st{8}$;
moreover, $0 < 8$ and $ip(0) = ip(8)$. The folding of $\N_1$ 
is shown at the bottom right of the figure.
It is clearly not trace-equivalent to $\A_1$ because it ``misses'' the trace 
$iabe$. The dashed event at the bottom right, which would correct this,
is not added by the algorithm because it is a successor of $8$. 

\begin{figure}[htbp]
\centering
\includegraphics[scale=0.55]{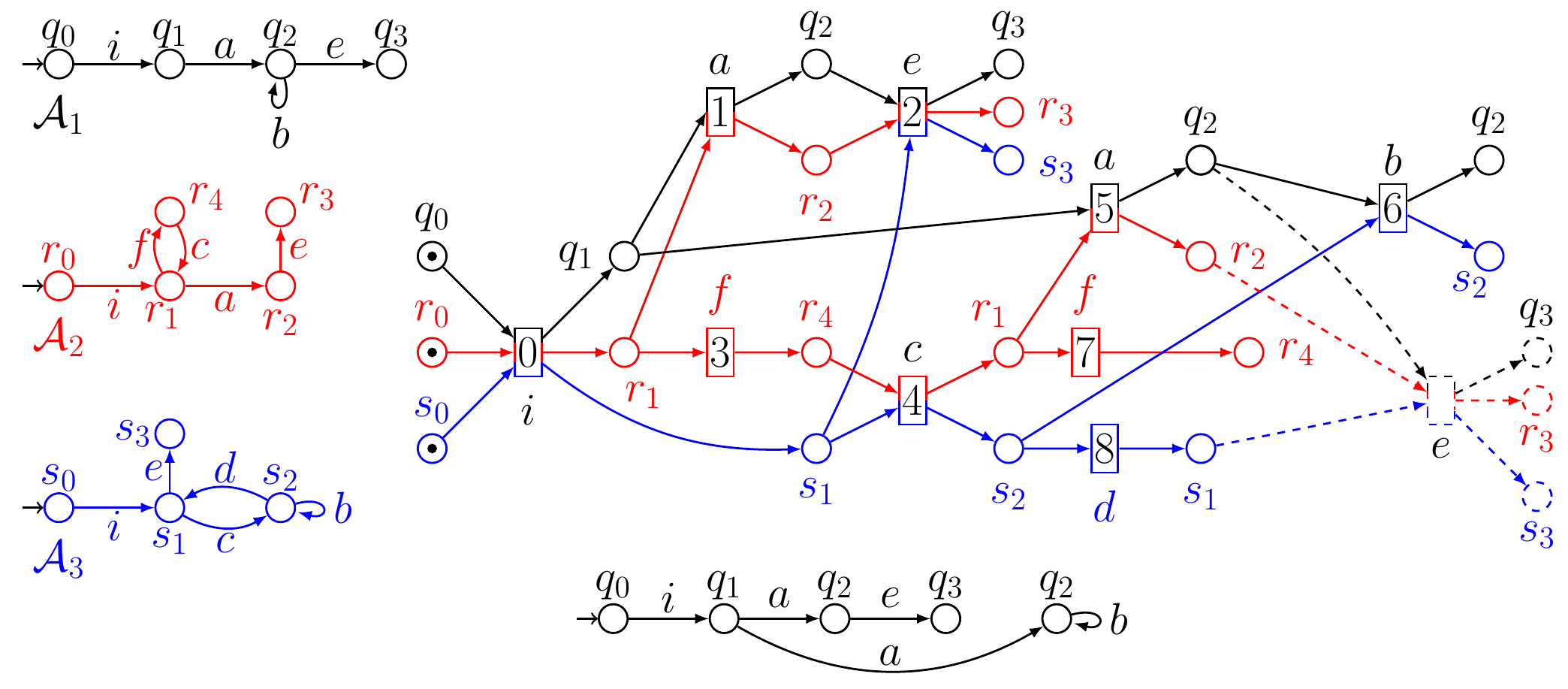}
\caption{Cut-off definition 2 produces an incorrect result on $\prod=\A_1 \parallel \A_2 \parallel \A_3$}\label{fig:needconc}
\label{fig:concneeded}
\end{figure}

\section{The Solution}

Intuitively, the reason for the failure of our second attempt
on the example of Figure \ref{fig:concneeded} is that $\A_1$ can only execute $iabe$ if $\A_2$ and $\A_3$ 
execute $\mathit{ifcd}$ first. However, when the algorithm observes that the markings before and after 
the execution of $\mathit{ifcd}$ are identical, it declares $8$ a cut-off event, and so it cannot ``use'' it to construct event $e$.
So, on the one hand, $8$ should not be a cut-off event. But, on the other hand, {\em some} event of the trace 
$\mathit{i}(\mathit{fcd})^\omega$ must be declared cut-off, otherwise the algorithm does not terminate. 

The way out of this dilemma is to introduce {\em cut-off candidates}.
If an event is declared a cut-off candidate, the algorithm does not add any of its 
successors, just as with regular cut-offs. However, cut-off candidates may stop being candidates if the addition of a new event {\em frees} them. 
(So, an event is a cut-off candidate {\em with respect to the current branching process}.) A generic unfolding procedure using these ideas is given in Algorithm~\ref{algo:unfolding2}, where 
$Ext(\N,\co,\tco)$ denotes the possible extensions of $\N$ that do not have any 
event of $\co$ or $\tco$ in their past.
Assuming suitable definitions of cut-off candidates and freeing, the algorithm
would, in our example, declare event $8$ a cut-off candidate, momentarily stop adding any of its successors, but later free event $8$ when event $5$ is discovered.

\begin{algorithm}[htbp]
\begin{algorithmic}
\State let $\N$ be the unique branching process of $\prod$ without events; let $\co=\emptyset$ and $\tco=\emptyset$
\State {\bf While} $Ext(\N,\co,\tco)\neq\emptyset$ {\bf do}
\State\hspace*{0.4cm} choose $e$ in $Ext(\N,\co,\tco)$ according to the search strategy
\State\hspace*{0.4cm} {\bf If} $e$ is a cut-off event {\bf then} let $\co=\co\cup\{e\}$
\State\hspace*{0.4cm} {\bf Elseif} $e$ is a cut-off candidate of $\N$ {\bf then} let $\tco=\tco\cup\{e\}$
\State\hspace*{0.4cm} {\bf Else for every} $e' \in \tco$ {\bf do}
\State\hspace*{0.8cm} {\bf If} $e$ frees $e'$ {\bf then} $\tco=\tco\setminus \{e'\}$
\State\hspace*{0.4cm} extend $\N$ with $e$
\State {\bf For every} $e \in \co$ with companion $e'$ {\bf do} merge $[M(e)_i]_\equiv$ and $[M(e')_i]_\equiv$
\end{algorithmic}
\caption{Unfolding procedure for a product $\prod$.}
\label{algo:unfolding2}
\end{algorithm}

The main contribution of our paper is the definition of a correct notion of cut-off candidate for the projection problem. 
We shall declare event $e$ a  cut-off candidate if $e$ is not an interface event,
and $\N$ contains a companion $e' < e$ such that $\st{e'} = \st{e}$, $ip(e)=ip(e')$, and, additionally, no interface event $e''$ of $\N$ is concurrent with $e$ without being concurrent with $e'$.
As long as this condition holds, the successors of $e$ are put ``on hold''. In the example of Figure \ref{fig:concneeded}, if the algorithm first adds events $0$, $3$, $4$, and $8$, then event $8$ becomes a cut-off candidate with $0$ as companion. However, the addition of the interface event $5$ frees event $8$, because $5$ is concurrent with $8$ and not with $0$.

However, we are not completely done yet. 
The parallel composition at the left of Figure~\ref{fig:strong} gives an example in which 
even with this notion of cut-off candidate the result is still wrong.
$\A_1$ is the interface.
One branching process is represented at the top right of the figure. 
Event 3 (concurrent with 1) is a cut-off candidate with 2 (concurrent with 1, 4, and 5) as companion.
This prevents the lower dashed part of the net to be added.
Event 6 is cut-off with 1 as companion.
This prevents the upper dashed part of the net to be added.
The refolding obtained then (bottom right) does not contain the word $abcb$.

\begin{figure}[htbp]
\centering
\includegraphics[scale=0.55]{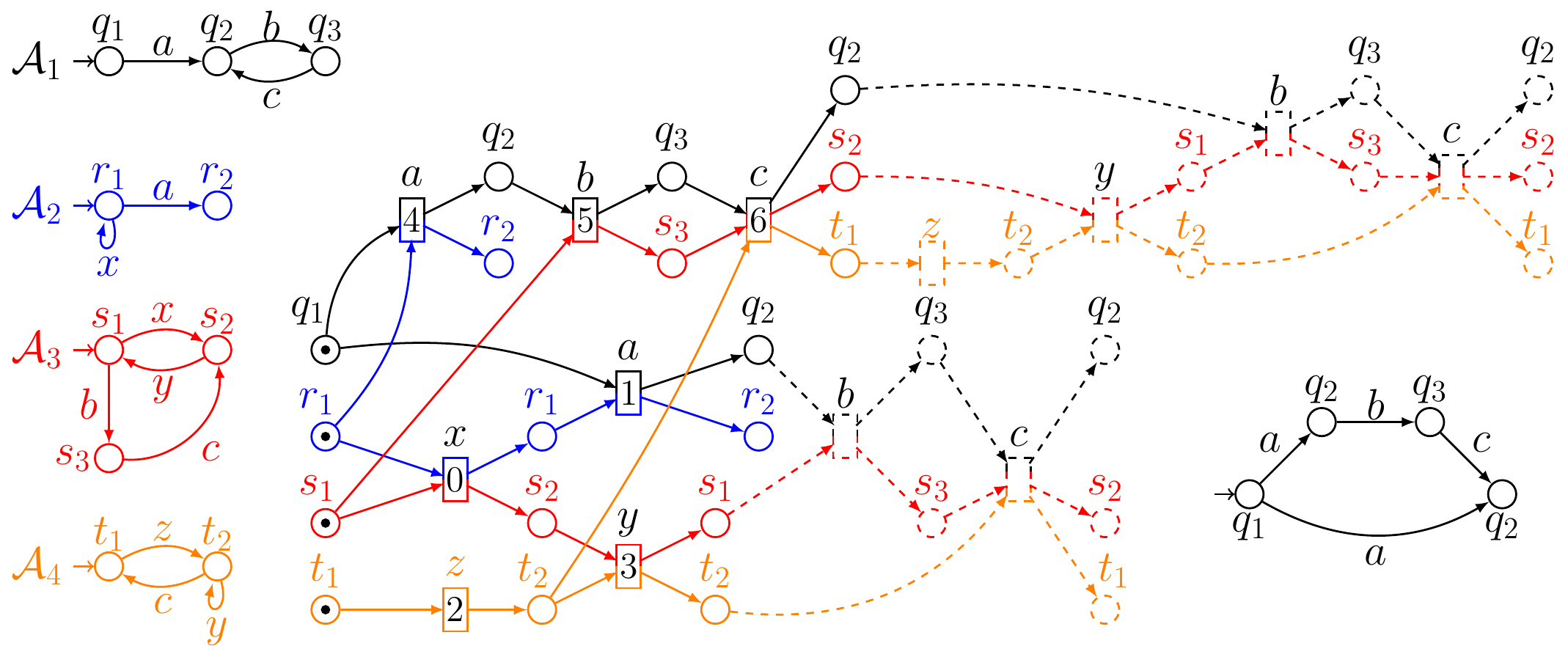}
\caption{An example illustrating the use of strong causality}\label{fig:strong}
\end{figure}

If we wish a correct algorithm for all strategies, we need a final touch: replace the condition $e' < e$ by $e' \scause e$, where $\scause$ is the {\em strong causal relation}:
\begin{definition}
Event $e'$ is a {\em strong cause} of event $e$, denoted by $e' \scause e$,
if $e' < e$ and $ b' < b$ for every $b \in \ma{e} \setminus \ma{e'}, b' \in \ma{e'} \setminus \ma{e}$. 
\end{definition}

Using this definition, event 3 is no longer a cut-off candidate in the branching process of Figure~\ref{fig:strong} as it is not in strong causal relation with its companion 2 (because the $t_2$-labelled condition just after 2 belongs to $\ma{2}\setminus\ma{3}$ and is not causally related with the $r_1$-labelled condition just after $0$ which belongs to $\ma{3}\setminus\ma{2}$).

The two following lemma give properties of the strong causal relation that will be useful to prove our main result (Theorem~\ref{thm:correctnessacceptable}).

\begin{lemma}
\label{lem:causal}
Every infinite chain $e_1 < e_2 < e_3 \cdots$ of events of a branching process
contains a strong causal subchain $e_{i_1} \scause e_{i_2} \scause e_{i_3} \cdots$. 
\end{lemma}
\begin{proof}
Let $E=\{e_1,e_2,\dots\}$. Say that a component $\A_j$ of $\prod$ participates
in an event $e$ if it participates in the transition labelling $e$. We partition the (indices of the)
components into the set $S$ of indices $j$ such that $\A_j$ participates in finitely many 
events of $E$, and $\bar{S} = \{1,\ldots,n\} \setminus S$. We say that the LTS $\A_j$ has 
{\em stabilized} at event $e_k$ in the chain if $\A_j$ does not participate in any event $e \geq  e_k$. 
Let $e_{\alpha}$ be any event of $E$ such that all LTSs of $S$ 
have stabilized before $e_{\alpha}$. We claim that there exists $e_\gamma$ in $E$ such that
$e_{\alpha} \scause e_{\gamma}$. Since clearly all LTSs of $S$ have also stabilized before $e_\gamma$,
A repeated application of the claim produces the desired subsequence. The claim itself
is proved in two steps:
\begin{itemize}
\item[(1)] There exists $e_{\beta} > e_{\alpha}$ in $E$ such that $M(e_{\beta})_k\neq M(e_{\alpha})_k$
for every $k\in\bar{S}$, (which implies $M(e_{\alpha})_k < e_\beta$ for every $k\in\bar{S}$).\\
The existence of $e_\beta$ follows from (1) the fact that all events of $E$ are causally related, and (2) the definition of $\bar{S}$, which implies for any $k\in\bar{S}$ the existence of an infinite subchain $e_{\ell_1}<e_{\ell_2}<\dots$ such that $M(e_{\ell_i})_k\neq M(e_{\ell_j})_k$ for every $i,j$.
\item[(2)] There exists $e_{\gamma} > e_{\beta}$ in $E$ such that $M(e_{\gamma})_k>e_\beta$ for every $k\in\bar{S}$.\\
Observe that if $e<M(e_i)_k$ for some $i$ and some $k$, then $e<M(e_j)_k$ for all $j>i$ (as $\forall i,j, \forall k,  M(e_i)_k\leq M(e_j)_k$).
Suppose that $e_\gamma$ does not exist.
Then there exists $k\in\bar{S}$ such that $M(e')_k\ngtr e$ for every $e'>e$.
As $k\in\bar{S}$, there exists, by definition, an infinite subchain $e<e_{\ell_1}<e_{\ell_2}\dots$ of $E$ such that $M(e_{\ell_i})_k\neq M(e_{\ell_j})_k$ for every $i,j$.
So for any of these $e_{\ell_i}$ there exists a $k$-event $e_{\ell_i}'$ such that 
$e_{\ell_i}'<e_{\ell_i}$ and $e_{\ell_i}'$ is concurrent with $e_{\ell_{i-1}}$.
Let $e_{\ell_i}''$ be an event on a path from $e_{\ell_i}'$ to $e_{\ell_i}$ and such that 
$b>e$ and $b'\ngtr e$ for some $b,b'\in \inp{e_{\ell_i}''}$ (the existence of such an event is ensured by the fact that $M(e_{\ell_i})_k\ngtr e$).
As $b>e$ we get $e_{\ell_i}''>e$ and thus $b''>e$ for every $b''\in \out{e_{\ell_i}''}$.
Hence, by the observation above, the set $\{k\in\bar{S}~:~M(e_{\ell_i})_k>e\}$ is strictly greater than the set $\{k\in\bar{S}~:~M(e_{\ell_{i-1}})_k>e\}$.
Since $\prod$ is finite, this contradicts the existence of $k\in\bar{S}$ such that 
$M(e')_k\ngtr e$ for every $e'>e$ in $E$. So the event $e_\gamma$ exists.
\end{itemize}
It follows immediately from (1) and (2) that $e_{\alpha} \scause e_{\gamma}$ (because for any $k,k'$, $M(e_\alpha)_k<e_\beta<M(e_\gamma)_{k'}$), and all LTSs of $S$ have stabilized before $e_\gamma$, and so the claim is proved.
\end{proof}

\begin{lemma}
\label{lem:conc}
If $e' \scause e$ and $\hat{e}$ is concurrent with both $e'$ and $e$, then 
$([e] \setminus [e']) \cap [\hat{e}] = \emptyset$.
\end{lemma}

\begin{proof}
Assume $e_1 \in ([e] \setminus [e']) \cap [\hat{e}]$. 

Then $e_1 \leq e$ and $e_1 \leq \hat{e}$. Since $e$ and $\hat{e}$ are concurrent,
we have $e \neq e_1 \neq \hat{e}$. So $e_1 < \hat{e}$, and so
there is a nonempty path $e_1 \prec b_1 \prec e_2 \prec b_2 \prec \ldots \prec e_k = \hat{e}$,
where $x \prec y$ denotes $y \in \out{x}$. Since $e$ and $\hat{e}$ are concurrent, there is 
a first condition $b_j$ in the path such that $b_j$ and $e$ are concurrent, and we have
$b_j \in \ma{e}$. Since $e_1 \notin [e']$, we have $b_j \notin \ma{e'}$. Since
$e' \scause e$, we have $b_j < b$ for every $b \in \ma{e'} \setminus \ma{e}$. In
particular, since there is at least one condition $b'$ such that $e'\prec b' < e$,
we have $b_j < b'$, and so $e' < b_j$. But then, since $b_j$ belongs to the path 
from $e_1$ to $\hat{e}$, we have $e' < b_j < \hat{e}$, contradicting that 
$e$ and $\hat{e}$ are concurrent.
\end{proof}

We are now in a position to provide adequate definitions
for Algorithm~\ref{algo:unfolding2}.

\begin{definition}[Cut-off and cut-off candidate] 
\label{def:ccos}
Let $\Coni{\N}{e}$ denote the set of 
non cut-off interface events of $\N$ that are concurrent with $e$. An event $e$  
\begin{itemize}
\item is a cut-off if it is an $i$-event, and $\N$ contains an $i$-event $e'$ 
such that $\st{e}=\st{e'}$.
\item is a cut-off candidate of $\N$ if it is not an $i$-event, and $\N$ contains $e' \scause  e$ such that
$\st{e}=\st{e'}$, $ip(e')=ip(e)$, and $\Coni{\N}{e} \subseteq \Coni{\N}{e'}$. 
\item frees a cut-off candidate $e_c$ of $\N$ if $e_c$ is not a cut-off candidate of the branching process
obtained by adding $e$ to $\N$. 
\end{itemize}
\end{definition}

\begin{theorem}
\label{thm:correctnessacceptable}
Let $\prod = \A_1 \parallel \ldots \parallel \A_n$ with interface $\A_i$. The instance of Algorithm~\ref{algo:unfolding2} given by Definition \ref{def:ccos} terminates and returns a branching process $\N$ such that the folding $\S_i$
of $\N_i$ is a summary of $\prod$.
\end{theorem}

\begin{proof}
We first prove termination. Assume the algorithm does not terminate, i.e., it constructs an
infinite branching process $\N$.
Then there exists an infinite chain $e_1<e_2<...$ of causally related events in $\N$~\cite{Khomenko03}.
First remark that $C=\cup_{i=1}^\infty [e_i]$ cannot contain an infinite number of $i$-events:
if there is infinitely many $i$-event in $C$ one of them must be a cut-off (this is due to the finite number of global states in $\prod$) as all the $i$-events of $C$ are causally related there is a contradiction.
Hence, $C$ contains an infinite chain $\ocseq'$ of causally related events such that for any two events $e$ and $e'$ of $\ocseq'$ one has $M(e)_i=M(e')_i$.
From that, the finite number of possible global states in $\prod$ ensures that there exists an infinite subchain $\ocseq''$ of $\ocseq'$ such that for any two events $e$ and $e'$ of $\ocseq''$ one has $\st{e}=\st{e'}$. 
The finite number of possible global states in $\prod$ also ensures that in $\N$ there exists only a finite set of non-cut-off $i$-events.
So, there exists an infinite subchain $\ocseq'''$ of $\ocseq''$ such that for any two events $e$ and $e'$ of $\ocseq'''$ one has $\Coni{\N}{e}=\Coni{\N}{e'}$.
Finally, by Lemma~\ref{lem:causal} there exists two events $e$ and $e'$ of $\ocseq'''$ such that $e'\scause e$.
Then, $e$ is a cut-off candidate of $\N$, which is in contradiction with the infiniteness of $\ocseq'''$ and so with the existence of $e_1<e_2<\dots$.
The termination of Algorithm~\ref{algo:unfolding2} is thus proved.

\vspace{0.2cm}
Now we prove $\Tr{\S_i}=\Tr{\prod}|_{\Sigma_i}$. As in the proof of Theorem 
\ref{th:interfacefair}, we extend the mapping $\st{}$ to conditions, and to equivalence classes of
conditions of $\N_i$. 

$\Tr{\S_i} \subseteq \Tr{\prod}|_{\Sigma_i}$. The proof of this part is identical to
that of Theorem \ref{th:interfacefair}: since the folding $\S_i$ is completely determined by the cut-offs
that are $i$-events, and the definition of these cut-offs in Definition 2 and Definition 5
coincide, the same argument applies. 

$\Tr{\prod}|_{\Sigma_i} \subseteq \Tr{\S_i}$. The proof has the same structure 
as the proof of Theorem \ref{th:interfacefair}, but with a number of important changes. 

Let $\trace$ be a (finite or infinite) trace of $\prod$.
We prove that there exists a trace $\trace^i$ of $\S_i$ such that $\trace^i=\trace|_{\Sigma_i}$.
For that we prove that for every history $h$ of $\prod$ there exists a history $h^i$ of $\S_i$ 
such that $\trace(h^i)=\trace(h)|_{\Sigma_i}$.

As in Theorem \ref{th:interfacefair}, we use the notion of a succinct histories. However, we need
to strengthen it even more. Let $\nu=\s_1\s_2\s_3 \ldots$ be a (finite or infinite) sequence of
global states of $\prod$, and let $H(\nu)$ be the (possibly empty) set of 
succinct histories $h$ with $i$-decomposition
$h_1h_2h_3 \ldots$ such that $\s^0 \by{h_1} \s_1 \by{h_2} \s_2 \by{h_3}\cdots$. 
We say that a history $h_s \in H(\nu)$ with $i$-decomposition $h_{1s}h_{2s}h_{3s}\ldots$
is {\em strongly succinct} if for every history $h \in H(\nu)$ with $i$-decomposition
$h_1h_2h_3 \ldots$ we have $|h_{js}|\leq |h_j|$ for every $j$.
If $h_1 \ldots h_jh_{j+1}h_{j+2} \ldots$ is succinct, $\s_{j-1} \by{h'_j} \s_j$, and 
$|h_{j}|\leq |h'_j|$, then $h_1 \ldots h'_jh_{j+1}h_{j+2} \ldots$ is also succinct. Therefore,
if $H(\nu)$ is nonempty then it contains at least one strongly succinct history.

As in Theorem \ref{th:interfacefair}, we prove by induction a result implying the one we need.
For every (finite or infinite) strongly succinct history of $\prod$ with $i$-decomposition $h=h_1h_2h_3 \ldots$
there exists $h_1^ih_2^ih_3^i\dots$ such that for every $k$:
\begin{enumerate}
\item[(a)] $H_k^i=h_1^i\dots h_k^i$ is a history of $\S_i$ such that $\trace(H_k^i)=\trace(h_1\dots h_k)|_{\Sigma_i}$.
\item[(b)] There exists a configuration $C_k$ of $\N$ that contains no cut-offs and such that
$[\ma{C_k}_i]_\equiv$ is the state reached by $H_k^i$. 
\item[(c)] If $k\neq 0$, then there exists an $i$-event $e_k$ such that $C_k=[e_k]$. 
\end{enumerate}
\noindent (The first two claims are as Theorem \ref{th:interfacefair}, while the third one is new.)

\noindent{\bf Base case.}
If $k=0$, then $H_k^i$ is the empty history of $\S_i$ and $C_k = \emptyset$.

\noindent{\bf Inductive step.}
The initial part of the inductive step is identical to that of Theorem \ref{th:interfacefair}. 
Let $H_{k+1}$ be the prefix of $h_1h_2h_3\dots$ with $i$-decomposition $H_{k+1}=h_1 \ldots h_k h_{k+1}$ (it is a strongly succinct history).
Then $H_k = h_1 \ldots h_k$ is strongly succinct with $i$-decomposition $h_1 \ldots h_k$. 
By induction hypothesis $H_k^i$, some configuration $C_{k}$,
and, if $k\neq 0$, some event $e_{k}$ satisfy the conditions above. 

Let $o_{k+1}=e_1 \dots e_m$, where $m = |h_{k+1}|$,  be the only sequence of events 
whose labelling is $h_{k+1}$ and can occur in the order of the sequence from the marking $\ma{C_{k}}$ 
(this sequence always exists by the properties of $C_{k}$). 
Two cases are possible:

\vspace{0.1cm}
\noindent 1. $o_{k+1}$ contains no cut-off.\\
The proof of this case is as in Theorem \ref{th:interfacefair}. Part (c) follows because in Theorem \ref{th:interfacefair}
we choose $C_{k+1}$ as $C_{k} \cup \{e_1, \ldots, e_m\}$, which, since $e_j \leq e_m$ for every $j \in \{1..m\}$, implies $C_{k+1} = [e_m]$.

\vspace{0.1cm}
\noindent 2. $o_{k+1}$ contains some cut-off event.\\
In Theorem \ref{th:interfacefair} we used the following argument: since $e_m$ is the only $i$-event of
$o_{k+1}$, and cut-offs must be $i$-events, $e_m$ is a cut-off. This argument is no longer valid,
because in Definition \ref{def:ccos} non-$i$-events can also be cut-offs. So we prove 
that $e_m$ is a cut-off in a different way.

Let $e$ be a cut-off of $o_{k+1}$, and let $e'$ be its companion. 
We prove that, due to the minimality of $h_{k+1}$ in the definition of strong succinctness,
we have $e=e_m$.

Assume $e \neq e_m$. Since $e_m$ is the unique $i$-event of $o_{k+1}$,
$e$ is not an $i$-event. So, by Definition \ref{def:ccos}, it is an event that became a cut-off candidate
and was never freed. %In particular, it is not an $i$-event. 

We consider first the case in which $C_k$ is the empty configuration (i.e. $k=0$).
In this case, consider a permutation $j_1j_2j_3$ of $o_{k+1}$ in which $j_1$ contains the events of $[e']$, $j_2$ contains the events of $[e]\setminus [e']$, and $j_3$ contains the rest of the events.
Since $\st{e}=\st{e'}$, $H_k\lambda(j_1j_3)=\lambda(j_1j_3)$ is also a history of $\prod$.
Since $|j_1j_3|<|o_{k+1}|$ this contradicts the minimality of $h_{k+1}$.

If $C_{k}$ is nonempty, then the $i$-event $e_{k}$ in part (c) of the induction hypothesis exists.
We consider the events $e$ and $e_{k}$. Since $e_{k}$ is an $i$-event but $e$ is not,
we have $e \neq e_{k}$. Since there is an occurrence sequence that contains
both $e$ and $e_{k}$, the events are not in conflict. Moreover, since in this occurrence sequence
$e$ occurs after $e_{k}$, we have that $e$ is not a causal predecessor of $e_{k}$ either.
So there are two remaining cases, for which we also have to show that they lead to a contradiction: 

\vspace{0.1cm}
\noindent (b1) $e_{k} < e$. Let $e'$ be the companion of $e$. By the definition of cut-off
candidate, we have $ip(e)=ip(e')$. Since $e_{k}$ is an $i$-event and $e_{k} < e$, we have 
$e_{k} < ip(e)$, and so $e_{k} < e' \scause e$. Consider the permutation $j_1j_2j_3$ of
$o_{k+1}$ in which $j_1$ contains the events of $[e'] \setminus [e_{k}]$,
$j_2$ contains the events of $[e] \setminus [e']$, and $j_3$ the rest
of the events. Since $\st{e}=\st{e'}$, $H_k\lambda(j_1j_3)$ is also a history of $\prod$. Since
$|j_1j_3| < |o_{k+1}|$, this contradicts the minimality of $h_{k+1}$.

\vspace{0.1cm}
\noindent (b2) $e_{k}$ and $e$ are concurrent. We handle this case by means 
of a sequence of claims. 
\begin{itemize}
\item[(i)]  Let $e'$ be the companion of $e$. The events $e'$ and $e_{k}$ are concurrent. \\
Follows from the fact that $e_{k}$ is an $i$-event
and $\Coni{\N}{e} \subseteq \Coni{\N}{e'}$ by the definition of cut-off candidate.
\item[(ii)] $([e] \setminus [e']) \cap [e_{k}]=\emptyset$.\\
Follows from Lemma \ref{lem:conc}, assigning $\hat{e} := e_{k}$.
\item[(iii)] $h_{k+1}$ is not minimal, contradicting the hypothesis.\\
By (ii), the sets $[e_{k}]$ and $[e] \setminus [e']$ are disjoint.
So every event of $[e] \setminus [e']$ belongs to $o_{k+1}$.
Consider the permutation $j_1j_2j_3$ of
$o_{k+1}$ in which $j_1$ contains the events that do not belong to $[e']$,
$j_2$ contains the events of $[e] \setminus [e']$, and $j_3$ the rest. 
Since $\st{e}=\st{e'}$, $H_k\lambda(j_1j_3)$ is also a history of $\A$, and since
$|j_1j_3| < |o_{k+1}|$ the sequence $h_{k+1}$ is not minimal.
\end{itemize}

Since all cases have been excluded, and so we have $e=e_m$, i.e., the $i$-event 
$e_m$ is the unique cut-off of $o_{k+1}$. Now we can reason as in Theorem \ref{th:interfacefair}. We have that
$o_{k+1}$ is a sequence of events from $\N$ whose last event is a cut-off, and 
the marking reached by $o_{k+1}$ is $\ma{e_m}$.
By the definition of folding, $\S_i$ has an execution $h_{i,k+1}$ from the state 
$[\ma{C_{k}}_i]_\equiv$ to the state $[\ma{e_m}_i]_\equiv$ such that 
$\trace(h_{i,k+1})=\trace(h_{k+1})|_{\Sigma_i}$. This allows to take $h^i_{k+1}=h_{i,k+1}$. 
We choose $C_{k+1}= [e_m']$, 
where $e_m'$ is the companion of $e_m$ and then, obviously $e_{k+1} = e_m'$. 
Since $e_m'$ is not a cut-off, $C_{k+1}$ contains no cut-offs. Moreover, since the marking 
reached by $o_{k+1}$ is $\ma{e_m}$, we have that $[\ma{C_{k+1}}_i]_\equiv$ is the state reached by 
$H_{k+1}^i$.
\end{proof}

\section{Implementation and Experiments}

\label{sec:exp}

As an illustration of the previous results, we report in this section on an implementation of 
Algorithm~\ref{algo:unfolding2}. All programs and data
used are publicly available.%
\footnote{\url{http://www.lsv.ens-cachan.fr/~schwoon/tools/mole/summaries.tar.gz}}

\subsection{Implementation}

We implemented Algorithm~\ref{algo:unfolding2} by modifying the
unfolding tool \textsc{Mole}~\cite{Mole}. The input of our tool is
the Petri net representation of a product $\prod$ in which every place
is annotated with the component it belongs to. Most of the infrastructure
of \textsc{Mole} could be re-used, in particular the existing implementation
contains efficient algorithms and data structures \cite{Esparza96}
for detecting new events
of the unfolding (the so-called possible extensions), computing the marking
$\st{e}$ of an event, etc.

The main work therefore consisted in determining cut-off candidates and
the ``freeing'' condition of Definition~\ref{def:ccos}.
For this, we introduce a \emph{blocking} relation between events: we write
$e'\blocks e$ if $e' \scause  e$, $\st{e}=\st{e'}$, $ip(e)=ip(e')$, and
$\Coni{\N}{e} \subseteq \Coni{\N}{e'}$, in other words $e$ is a cut-off
candidate because of $e'$; let $\rblocks{e}:=\{\,e'\in\N \mid e'\blocks e\,\}$.
Notice that $\rblocks{e}\subseteq\past{e}$.
Therefore, an over-approximation of this set can be computed when
$e$ is discovered as a possible extension, by checking all its causal
predecessors. When $\N$ is expanded, $\rblocks{e}$ can only decrease
because adding an event may lead to a violation of the condition
$\Coni{\N}{e}\subseteq\Coni{\N}{e'}$.

The blocking relation requires two principal, interacting additions
to the unfolding algorithm:
\begin{itemize}
\item[(i)] a traversal of $\past{e}$ collecting information about
  the `cut' $\ma{e}$;
\item[(ii)] computing the concurrency relation between events.
\end{itemize}

For (i), we modify the way \textsc{Mole} determines $\st{e}$:
it performs a linear traversal of $\past{e}$, marking all conditions
consumed and produced by the events of $\past{e}$, thus obtaining $\ma{e}$.
We extend this linear traversal with Algorithm~\ref{alg:walk}, which computes
$\mathit{cut}=\ma{e}$, allowing to directly determine the conditions
$\st{e}=\st{e'}$
and $ip(e)=ip(e')$. Moreover, every condition $b$ becomes annotated with
a set $\ind{b}:=\{\,j\mid b\le \ma{e}_j\,\}$.
This, together with $\ma{e}$ and $\ma{e'}$, allows
to efficiently determine whether $e'\scause e$ holds. Notice that if the
number of components in $\prod$ is ``small'', the operations on $\ind{b}$
can be implemented with bitsets. Thus, the additional overhead of
Algorithm~\ref{alg:walk} with respect to the previous algorithm can be
kept small.

\begin{algorithm}[htbp]
\begin{algorithmic}
\State let $\N$ be the current branching process and $e$ its latest extension
\State set $\worklist:=\past{e}$ and $\mathit{cut}:=\emptyset$
\State for all conditions $b$, let $b$ unmarked and $\ind{b}:=\emptyset$
\While{$\worklist\ne\emptyset$}
\State remove a $<$-maximal element $e$ from $\worklist$
\State add all unmarked conditions $b\in\out{e}$ to $\mathit{cut}$ and set $\ind{b}:=\{i(b)\}$
\State $I:=\bigcup_{b\in\out{e}}\ind{b}$;
\State mark all conditions $b\in\inp{e}$ and set $\ind{b}:=I$
\EndWhile
\State add all unmarked initial conditions $b$ to $\mathit{cut}$ and set $\ind{b}:=\{i(b)\}$
\end{algorithmic}
\caption{Traversal of $\past{e}$ for efficiently determining $\rblocks{e}$,
where $i(b)$ denotes the component to which condition $b$ belongs.}
\label{alg:walk}
\end{algorithm}

Concerning (ii), we are interested in determining the sets $\Coni{N}{e}$
for all events $e$. We make use of the facts that:
\begin{itemize}
\item \textsc{Mole} already determines, for every condition $b$, a set
  of other conditions $\parb{b}$ that are concurrent with $b$. When
  the $\N$ is extended with event $e$, it computes the set
  $I:=\bigcup_{b\in\inp{e}}\parb{b}$ and sets
  $\parb{b'}=I\cup\out{e}\setminus\{b'\}$ for every $b'\in\out{e}$.
\item Two events $e,e'$ of $\N$ are concurrent iff their inputs $\inp{e}$
  and $\inp{e'}$ are disjoint and pairwise concurrent. Thus, when $e$ is
  added, this relation can be checked by marking the events in $I$ and
  checking whether $I$ includes $\inp{e'}$. Thus, $\Coni{N}{e}$ can be
  obtained with small overhead w.r.t.\ the existing implementation.
\item At the same time, we can easily determine whether the addition
  of an event $e$ should lead to the removal of some event $e'$ from
  $\rblocks{e''}$; if this causes $\rblocks{e''}$ to become empty,
  $e''$ is freed.
\end{itemize}

\subsection{Experimental results}

We tested our implementation on well-known benchmarks
used widely in the unfolding literature,
see for example~\cite{Cor96,Esparza96,Khomenko03}.
The input is the set of components $\A_1, \ldots, \A_n$,
which are converted into an equivalent Petri net. 
All reported times are on a machine with a 2.8 MHz Intel CPU and
4~GB of memory running Linux.
For each example, we
also report the number of events
(including cut-offs) in the prefix (Events),
the number of states
in the resulting summary $\S_i$ ($|\S_i|$), the size of a minimal deterministic
automaton for a summary (Min), and
the number of reachable
markings (Markings, taken from \cite{RoemerPhd} where
available, and computed combinatorially for \textsc{DpSyn}).

The experiments are summarized in Table~\ref{tab:expfull}. We used
the following families of examples~\cite{Cor96}:
the \textsc{CyclicC} and \textsc{CyclicS} families are a model of
Milner's cyclic scheduler with $n$ consumers and $n$ schedulers; in
one case we compute the folding for a consumer, in the other for a scheduler.
The \textsc{Dac} family represents a divide-and-conquer computation.
\textsc{Ring} is a mutual-exclusion protocol on a token-ring. The tasks
are not entirely symmetric, we report the results for the first.
Finally, \textsc{Dp}, \textsc{Dpsyn}, and \textsc{Dpd} are variants of
Dining Philosophers. In \textsc{Dp},
philosophers take and release forks one by one, whereas in \textsc{Dpsyn}
they take and release both at once. In \textsc{Dpd}, 
deadlocks are prevented by passing a dictionary.

\begin{table}[ht]
\begin{center}
\setlength\tabcolsep{5pt}
\begin{tabular}{lrrrrr}
\midrule
Test case & Time/s & Events & $|\S_i|$ & Min. & Markings \\
\toprule
\textsc{CyclicC}(6)  &  0.04  &   426 & 5 & 2 &    639 \\
\textsc{CyclicC}(9)  &  0.17  &  3347 & 5 & 2 &   7423 \\
\textsc{CyclicC}(12) &  4.04  & 26652 & 5 & 2 &  74264 \\
\midrule
\textsc{CyclicS}(6)  & 0.05 &   303 & 11 & 5 &    639 \\
\textsc{CyclicS}(9)  & 0.12 &  2328 & 11 & 5 &   7423 \\
\textsc{CyclicS}(12) & 2.38 & 18464 & 11 & 5 &  74264 \\
\midrule
\textsc{Dac}(9)  & 0.02 &  86 & 4 & 4 &   1790 \\
\textsc{Dac}(12) & 0.03 & 134 & 4 & 4 &  14334 \\
\textsc{Dac}(15) & 0.03 & 191 & 4 & 4 & 114686 \\
\midrule
\textsc{Dp}(6)  & 0.06 &    935 & 20 & 4 &   729 \\
\textsc{Dp}(8)  & 0.22 &   5121 & 28 & 4 &  6555 \\
\textsc{Dp}(10) & 2.23 &  31031 & 36 & 4 & 48897 \\
\midrule
\textsc{Dpd}(4) &  0.10 &   2373 & 114 & 6 &    601 \\
\textsc{Dpd}(5) &  0.71 &  23789 & 332 & 6 &   3489 \\
\textsc{Dpd}(6) & 17.68 & 245013 & 903 & 6 &  19861 \\
\midrule
\textsc{Dpsyn}(10) & 0.02 &  176 & 2 & 2 &     123 \\
\textsc{Dpsyn}(20) & 0.07 &  701 & 2 & 2 &   15127 \\
\textsc{Dpsyn}(30) & 0.26 & 1576 & 2 & 2 & 1860498 \\
\midrule
\textsc{Ring}(5) & 0.07 &   511 &  53 & 10 &   1290 \\
\textsc{Ring}(7) & 0.12 &  3139 & 101 & 10 &  17000 \\
\textsc{Ring}(9) & 0.93 & 16799 & 165 & 10 & 211528 \\
\end{tabular}
\caption{More experimental results}
\label{tab:expfull}
\end{center}
\end{table}

In all cases except one (\textsc{Dpd}) our algorithm needs clearly 
fewer events than there are reachable markings;
in some families (\textsc{Dac},
\textsc{Dpsyn}, \textsc{Ring}) there are far fewer events.
A comparison of \textsc{Dp} and \textsc{Dpsyn} is instructive. In \textsc{Dp},
neighbours can concurrently pick and drop forks. Intuitively, this leads
to fewer cases in which the condition $\Coni{\N}{e}\subseteq\Coni{\N}{e'}$
for cut-off candidates is satisfied. On the other hand, in \textsc{Dpsyn} both forks are
picked and dropped synchronously, and so no event in $\A_i$ is concurrent to any
event in the neighbouring components, making the unfolding procedure much
more efficient.

\section{Extensions: Divergences and  Weights}

We conclude the paper by showing that our algorithm can be extended to handle more complex semantics than traces.
Indeed, the divergences of the system can be captured by the summaries, as well as the minimal weights of the finite traces from $\Tr{\prod}|_{\Sigma_i}$ when $\A_1\dots\A_n$ are weighted systems.

\subsection{Divergences}
We first extend our algorithm so that the summary
also contains information about {\em divergences}.
Intuitively, a divergence is a finite trace of the interface after which
the system can ``remain silent'' forever.

\begin{definition}
Let $\A_1, \ldots, \A_n$ be LTSs with interface $\A_i$. A 
{\em divergence} of $\A_i$ is a finite trace
$\sigma \in \Tr{\A_i}$ such that $\sigma = \tau_{|{\Sigma_i}}$ for some
infinite trace $\tau \in \Tr{\prod}$.
A {\em divergence-summary} is a pair $(\S_i, D)$, where $\S_i$ is a summary 
and $D$ is a subset of the states of $\S_i$ such that $\sigma \in \Tr{\S_i}$
is a divergence of $\A_i$ if{}f some realization of $\sigma$ in $\S_i$ 
leads to a state of $D$. 
\end{definition}

We define the set of divergent conditions of the output of Algorithm~\ref{algo:unfolding2}, and show that it is a correct choice for the set $D$.

\begin{definition}
Let $\N$ be the output of Algorithm~\ref{algo:unfolding2}. A condition $s$ of $\N_i$ is {\em divergent}
if after termination of the algorithm there is $e \in {\it coc}$ with 
companion $e'$ such that $s$ is concurrent to both $e$ and $e'$. We denote
the set of divergent conditions by ${\it DC}$.
\end{definition}

\begin{theorem}
\label{thm:divergences}
A finite trace $\sigma \in \Tr{\S_i}$ is a divergence of $\A_i$ 
if{}f there is a divergent condition $s$ of $\N_i$ such that some realization
of $\sigma$ leads to $[s]_\equiv$. Therefore, $(\S_i, [{\it DC}]_\equiv)$ is a divergence-summary.
\end{theorem}

\begin{proof}
$(\Rightarrow)$
Assume that $\sigma$ is a divergence of $\A_i$.
By the definition of a divergence, there exists $\tau\in\Tr{\prod}$ such that $\tau|_{\Sigma_i}=\sigma$ and $\tau$ is infinite.
So there exists a strongly succinct history $h$ of $\A$ such that $\trace(h)=\tau$. Denote by $e_i$ the last i-event of $h$.
The proof of Theorem~\ref{thm:correctnessacceptable} guarantees the existance
of an i-event $e_i'$ in $\N$ which is not a cut-off and satisfies the following two properties: $\st{e_i}=\st{e_i'}$, and there exists a realisation of $\sigma$ leading to $[s]_\equiv$, where $s=M(e_i)_i$.
As $\tau$ is infinite, the unfolding $\U$ of $\A$ contains an infinite occurrence sequence starting at $M(e_i)$ and containing no i-event. Since 
$\st{e_i}=\st{e_i'}$, another infinite sequence with the same labelling and without $i$-events can occur from $M(e_i')$ in $\U$. By construction of $\N$, and since $e_i'$ is not a cut-off, a non-empty prefix of this second occurrence sequence appears in $\N$, and contains at least one cut-off candidate $e$.
So $e$ appears in some occurrence sequence without i-events starting at $M(e_i')$.
It follows that $e$ is either (1) concurrent with $e_i'$, or 
(2) a successor of $e_i'$ such that $ip(e)=M(e_i')_i$. 
Moreover, since $e$ is not an i-event, it is concurrent with $s=M(e_i')_i$.
It remains to show that the companion $e'$ of $e$ is also concurrent with $s$.
If (1) holds, i.e., if $e$ is concurrent with $e_i'$, then $e'$ is concurrent with $e_i'$ (and so with $s$) as well, because, by the definition of a cut-off candidate, we have $\Coni{\N}{e}\subseteq\Coni{\N}{e'}$.
If (2) holds, i.e., if $e>e_i'$, then we have $e'>e_i'$ for the same reason as in the case (b1) in the proof of Theorem~\ref{thm:correctnessacceptable}), and so $e'$ and $s$ are concurrent.

$(\Leftarrow)$
Consider a divergent condition $s$ of $\N_i$.
By the definition of a divergent condition there exist a cut-off candidate $e$ with companion $e'$ such that neither $e$ nor $e'$ are i-events, 
and both $e$ and $e'$ are concurrent with $s$.
Let $e_i$ be the i-event such that $M(e_i)_i=s$.
As $e$ is concurrent with $s$, it is either concurrent with $e_i$, or a successor of $e_i$ such that $ip(e)=M(e_i)_i$. We consider these two cases separately.

(1) $e$ is a successor of $e_i$ such that $ip(e)=M(e_i)_i$.  Then $e'$ is a successor of $e_i$ for the same reason as in case (b1) of Theorem~\ref{thm:correctnessacceptable}. So we have $[e_i]\subseteq[e']\subseteq[e]$.
Let $j_1$ be any occurrence sequence starting from $M(e_i)$ and containing exactly the events in $[e'] \setminus [e_i]$ (so $j_1$ contains no i-events).
Let $j_2$ be any occurrence sequence starting at $M(e')$ and containing exactly the events in $[e]\setminus [e']$ (so $j_2$ contains no i-events either).
As $\st{e}=\st{e'}$, there exists an occurrence sequence $j_2^1$ 
in $\U$ starting at $M(e)$ and such that $\trace(j_2^1)=\trace(j_2)$; moreover the last event $e^1$ of $j_2^1$ satisfies $\st{e^1}=\st{e}$.
So we can iteratively construct occurrence sequences $j_2^k$ for every $k>1$, each of them starting at $M(e^{k-1})$, satisfying $\trace(j_2^k)=\trace(j_2)$, 
and ending with an event $e^k$ satisfying $\st{e^k}=\st{e}$.
So the infinite occurrence sequence $j_1j_2j_2^1j_2^2\dots$ can occur in $\U$ from $M(e_i)$.

(2) $e$ is concurrent with $e_i$.  Then $e'$ is also concurrent with $e_i$, because the definition of a cut-off candidate requires $\Coni{\N}{e}\subseteq\Coni{\N}{e'}$. By Lemma~\ref{lem:conc} we have $[e_i]\cap([e]\setminus[e'])=\emptyset$. Let $j_1$ be any occurrence sequence starting from $M(e_i)$ and containing exactly the events in $[e'] \setminus [e_i]$ (so $j_1$ contains no i-events).

Given two arbitrary concurrent events $e_1, e_2$, let $M(e_1,e_2)$ be the unique marking reached by any occurrence sequence that fires exactly the events of $[e_1]\cup[e_2]$. Let $j_2$ be any occurrence sequence starting from $M(e_i,e')$ and containing exactly the events in $[e']\setminus [e]$ (so $j_2$ contains no i-events).
As $\st{e}=\st{e'}$ and $[e_i]\cap([e]\setminus[e'])=\emptyset$, there exists  an occurrence sequence $j_2^1$ in $\U$ starting at $M(e_i,e)$ and such that $\trace(j_2^1)=\trace(j_2)$; moreover the last event $e^1$ of $j_2^1$ satisfies $\st{e^1}=\st{e}$. So for every $k > 1$ we can iteratively
construct sequences $j_2^k$ starting from $M(e_i,e^{k-1})$ such that $\trace(j_2^k)=\trace(j_2)$ and ending with an event $e^k$ satisfying $\st{e^k}=\st{e}$. It follows that the infinite occurrence sequence $j_1j_2j_2^1j_2^2\dots$ can occur in $\U$ from $M(e_i)$.

So in both cases $\prod$ has an infinite execution $h'$ starting at $\st{e_i}$ and such that $\trace(h')|_{\Sigma_i}$ is empty.
Moreover, if some realization of $\sigma$ leads to $[s]_\equiv=M(e_i)_i$, the proof of Theorem~\ref{thm:correctnessacceptable} guarantees the existence of a history $h$ of $\prod$ reaching state $\st{e_i}$ and satisfying $\trace(h)|_{\Sigma_i}=\sigma$.
Taking $\tau=\trace(hh')$ concludes the proof.
\end{proof}

\subsection{Weights}
We now consider weighted systems, e.g parallel compositions of weighted LTS.
Formally, a weighted LTS $\A^w=(\A,c)$ consists of an LTS $\A=(\Sigma,S,T,\lambda,s^0)$ and a weight function $c:T\rightarrow \R$ associating a weight to each transition.
A \emph{weighted trace} of $\A^w$ is a pair $(\sigma,w)$ where $\sigma=a_1\dots a_k$ is a finite trace of $\A$ and $w$ is the minimal weight among the paths realizing $\sigma$, i.e:
$$w=\min_{\substack{
s_0\dots s_k\in S^{k+1}, s_0=s^0,\\
t_i=(s_{i-1},s_i)\in T, \lambda(t_i)=a_i}}\sum_{j=1}^kc(t_j).$$
We denote by $\Tr{\A^w}$ the set of all the weighted traces of $\A^w$.
The parallel composition $\prod^w=(\prod,\mathbf{c})=\A^w_1\mathbin{||_w}\cdots\mathbin{||_w}\A^w_n$ of the LTS $\A^w_1,\dots,\A^w_n$ is such that $\prod=\A_1||\dots||\A_n$ and the weight of a global transition $\t=(t_1,\dots,t_n)$ is: 
$$\mathbf{c}(\t)=\sum_{t_i\neq\star}c_i(t_i).$$

Similarly a \emph{weighted labelled Petri net} is a tuple $\N^w=(\N,c)$ where $\N=(\Sigma,P,T,F,\lambda,M_0)$ is a labelled Petri net and $c:T\rightarrow\R$ associates weights to transitions.
A weighted trace in $\N^w$ is a pair $(\sigma,w)$ with $\sigma$ a finite trace of $\N$ and $w$ the minimal weight of an occurrence sequence corresponding to $\sigma$, where the weight of an occurrence sequence is the sum of the weights of its transitions.
By $\Tr{\N^w}$ we denote the set of all the weighted traces of $\N^w$.

The branching processes of $\A_1^w||_w\dots||_w\A_n^w$ are defined as weighted labelled Petri nets like in the non-weighted case, where
each event is implicitly labelled by an action (as before) and a cost.
Given a finite set of weighted traces $W$ we define its restriction to alphabet $\Sigma$ as 
$$W|_{\Sigma}=\{\,(\sigma,w)\mathrel{:}\exists(\sigma',w')\in W, \sigma=\sigma'|_{\Sigma} \wedge w=\min_{\substack{(\sigma',w')\in W\\ \sigma'|_{\Sigma}=\sigma}} w'\,\}.$$
As in the non-weighted case we are interested in solving the following summary problem:
\begin{definition}[Weighted summary problem]
Given $\A^w_1,\dots,\A^w_n$, weighted LTSs with interface $\A^w_i$, compute a weighted LTS $\S_i^w$ satisfying $\Tr{\S_i^w}=\Tr{\prod^w}|_{\Sigma_i}$, where $\prod^w=\A_1^w||_w\dots||_w\A_n^w$.
\end{definition}

This section aims at showing that the approach to the summary problem proposed in the non-weighted case still works in the weighted one.
In other words, $\S^w_i$ can be obtained by computing a finite branching process $\N^w$ of $\prod^w$ (using Definition~\ref{def:ccos} of cut-off and cut-off candidates and Algorithm~\ref{algo:unfolding2}) and then taking the interface projection $\N^w_i$ of $\N^w$ on $\A_i^w$ and folding it.
The notion of interface projection needs to be slightly modified to take weights into account.
The conditions, events, and arcs of $\N^w_i$ are defined exactly as above, and the weight of an event $e$ of $\N^w_i$ is $c_i(e)=c([e])-c([e'])$ if the predecessor $e'$ of $e$ in $\N^w_i$ exists and $c_i(e)=c([e])$ else, where $c$ is the weight function of $\N^w$ and $c([e])=\sum_{e_k\in[e]}c(e_k)$, where $[e]$ is the past of $e$ in the weighted branching process $\N^w$.

\begin{theorem}
\label{thm:costs}
Let $\prod^w=\A_1^w||_w\dots||_w\A_n^w$ with interface $\A_i^w$.
The instance of Algorithm~\ref{algo:unfolding2} given by Definition~\ref{def:ccos} terminates and returns a weighted branching process $\N^w$ such that the folding $S_i^w$ of $\N_i^w$ is a weighted summary of $\prod^w$. 
\end{theorem}

\begin{proof}
The termination is granted by Theorem~\ref{thm:correctnessacceptable} as well as the fact that the weighted trace $(\trace,w)$ belongs to $\Tr{\S^w_i}$ if and only if, for some $w'$, the weighted trace $(\trace,w')$ belongs to $\Tr{\prod^w}_{|\Sigma_i}$.
It remains to show that for any $\trace$ such that $(\trace,w)\in\Tr{\S^w_i}$ and $(\trace,w')\in\Tr{\prod^w}_{|\Sigma_i}$ one has $w=w'$.
In the following we denote by $c_i$ the costs functions of $\S^w_i$ and $\N^w_i$, and by $\mathbf{c}$ the cost function of $\prod^w$.
Similarly we denote by $\lambda_i$ the labelling function of $\N^w_i$ and by $\lambda$ the labelling function of $\prod^w$.

$w'\leq w$.
This part of the proof is very close to the proof of the first inclusion of Theorem~\ref{th:interfacefair}.
Let $(\trace^i,w)$ be a finite weighted trace of $\S_i^w$. 
Then $[b^0]_\equiv \by{\trace^i} [b]_\equiv$ in $\S_i^w$ with $c_i(\trace^i)=w$,
where $[b^0]_\equiv$ is the initial state of $\S_i$, and $[b]_\equiv$ is some state of $\S_i$.
By the definition of folding, there exist $\tau^i_{1}, \ldots, \tau^i_{k}$ occurrence sequences of $\N_i$ and $(b_1',b_1), \ldots, (b_k',b_k)$ pairs of conditions of $\N_i$ such that
(1) $\trace^i = \lambda_i(\tau^i_{1}) \lambda_i(\tau^i_{2}) \ldots \lambda_i\tau^i_{k}$; (2) $b^0=b'_1$ and $b_k=b$; (3) $b_{j}' \by{\tau_{j}^i} b_j$ in $\N_i$ for every $j=1,\ldots, k$; 
(4) $b_{j-1} \equiv  b_j'$ for every $j\in\{1..k\}$;
and (5) $c_i(\tau^i_1) + \dots + c_i(\tau^i_k)=c_i(\trace^i)=w$.

By (3) and the definition of projection, we have
$\st{b_{j}'} \by{\tau_{j}} \st{b_j}$ in $\prod$ for some execution $\tau_j$ such that $\lambda_i(\tau^i_{j})=\lambda(\tau_j)|_{\Sigma_i}$ and $\mathbf{c}(\tau_j)= c_i(\tau_j^i)$: indeed, if $e$ and $e'$ are the input events of $b_j$ and $b'_{j}$, then $\st{b_j}$ is reachable from $\st{b'_{j-1}}$ by means of any execution $\tau_j$ corresponding to executing the events of $[e] \setminus [e']$, and any
such $\tau_j$ satisfies $\lambda_i(\tau^i_{j})=\lambda(\trace_j)|_{\Sigma_i}$ and $\mathbf{c}(\tau_j)= c_i(\tau_j^i)$. Moreover, by (4) we have $\st{b_{j-1}}=\st{b_j'}$. So we get 
$$\st{b_1'} \by{\tau_{1}} \st{b_2'} \by{\tau_{2}} \cdots \by{\tau_{k-1}}\st{b_{k}'} \by{\tau_{k}} \st{b_k}$$
\noindent By (1) and (2) we have $\st{b^0}\by{\tau_1 \ldots \tau_k} \st{b}$ in $\prod$, so
$\trace^i=\trace|_{\Sigma_i}\in\Tr{\prod}|_{\Sigma_i}$ with $\trace=\lambda(\tau_1)\ldots\lambda(\tau_k)$, and by (5) and the definition of a weighted trace $w'\leq \mathbf{c}(\trace)\leq \mathbf{c}(\tau_1)+\dots+\mathbf{c}(\tau_k)= c_i(\trace_1^i)+\dots+c_i(\trace_k^i)=w$.

$w\leq w'$.
This part of the proof is almost exactly the same as the proof of the second inclusion of Theorem~\ref{thm:correctnessacceptable} (considering finite traces only).
We describe here the few differences between these two proofs.
The main one is the definition of strongly succinct histories: instead of requiring $|h_{js}|\leq|h_j|$ we require $\mathbf{c}(h_{js})<\mathbf{c}(h_j)$, or $\mathbf{c}(h_{js})=\mathbf{c}(h_j)$ and $|h_{js}|\leq|h_j|$.
Then, as we are interested in weights, claim (a) of the induction hypothesis has the supplementary requirement that $c_i(H_k^i)=\mathbf{c}(h_1\dots h_k)$.
The base case is then the same, just remarking that the cost of the empty history is $0$ in both $\S_i^w$ and $\prod^w$. 
For the inductive step two things have to be done: (1) ensuring that when $o_{k+1}$ contains a cut-off it is necessarily $e_m$ and (2) ensuring the new part of claim (a) about weights.
For (1) just remark that in all cases $j_1j_3$ is such that $\mathbf{c}(j_1j_3)\leq \mathbf{c}(j_1j_2j_3)$ and $|j_1j_3|<|j_1j_2j_3|$ so the same arguments as previously can be used with the new definition of a strongly succinct history.
For (2) notice that when $e_m$ is a cut-off i-event, in the unfolding of $\prod^w$ the events that can occur from $M(e_m)$ and from $M(e_m')$ do not only have the same labelling: they in fact correspond to the exact same transitions of $\prod^w$ and so they also have the same weights.

Reusing this proof we have shown that the weighted trace $(\trace,w')$ of $\A^w$ is such that there exists a history $h^i$ of $\S_i^w$ such that $\trace_{|\Sigma_i}=\trace(h^i)$ and $c_i(h^i)=\mathbf{c}(\trace)=w'$.
So, by the definition of a weighted trace it comes directly that $w\leq c_i(h^i) = w'$.
\end{proof}

\section{Conclusions}

We have presented the first unfolding-based solution to the 
summarization problem for trace semantics. The final algorithm is 
simple, but its correctness proof is surprisingly subtle. 
We have shown that it can be extended (with minor modifications) to handle divergences and weighted systems.

The algorithm can also be extended to other semantics,
including information about failures or completed traces;
this material is not contained in the paper because, while laborious,
it does not require any new conceptual ideas. 

We conjecture that the condition $e' \scause e$ in the definition 
of cut-off candidate 
can be replaced by $e' < e$, if at the same time the algorithm is 
required to add events 
in a suitable order. Similar ideas have proved successful in the past 
(see e.g. \cite{Esparza96,Khomenko03}).

\bibliographystyle{plain}
\bibliography{bibliotech}

\end{document}